\documentclass[lettersize,journal,twoside]{IEEEtran}
\pdfoutput=1
\IEEEoverridecommandlockouts
\usepackage{cite}
\usepackage{amsthm}
\usepackage{amsmath,amssymb,amsfonts}
\usepackage{graphicx} 
\usepackage{float} 
\usepackage{epstopdf}
\usepackage{caption}
\usepackage{algorithmic}
\usepackage{textcomp}
\usepackage{xcolor}
\usepackage{subfigure}
\usepackage[justification=centering]{caption}
\usepackage{setspace} 

\usepackage{booktabs}
\usepackage{multirow}
\usepackage{tabularx}
\usepackage{longtable}
\usepackage{supertabular}
\def\BibTeX{{\rm B\kern-.05em{\sc i\kern-.025em b}\kern-.08em
    T\kern-.1667em\lower.7ex\hbox{E}\kern-.125emX}}

\newtheorem{definition}{Definition}
\newtheorem{theorem}{Theorem}

\newtheorem{lemma}{Lemma}
\newtheorem{proposition}{Proposition}
\newtheorem{remark}{Remark}

\newcommand{\sx}{\mkern-5mu}

\begin{document}

\title{Fourth-Order Hierarchical Array: A Novel Scheme for Sparse Array Design Based on Fourth-Order Difference Co-Array}

\author{Si Wang and Guoqiang Xiao
\thanks{
Manuscript received xx August 2025.
This work was supported in part by the National Natural Science Foundation of China under Grant 62371400.
The associate editor coordinating the review of this manuscript and approving it for publication was xxx.
 (Corresponding author: Guoqiang Xiao.)}
\thanks{Si Wang and Guoqiang Xiao are with the College of Computer and Information Science, Southwest University, Chongqing 400715, China
 (e-mail:hw8888@email.swu.edu.cn; gqxiao@swu.edu.cn).}
}

\markboth{ A Novel Scheme for Sparse Array Design Based on Fourth-Order Difference Co-Array,~Vol.~xx, No.~x, August~2025}%
{WNAG \MakeLowercase{\textit{et al.}}: A Novel Scheme for Sparse Array Design Based on Fourth-Order Difference Co-Array}

\maketitle

\begin{abstract}
Conventional array designs based on circular fourth-order cumulant
typically adopt a single expression form of the fourth-order difference co-array (FODCA),
which limits the achievable degrees of freedom (DOFs) and neglects the impact of mutual coupling among physical sensors.
To address above issues, this paper proposes a novel scheme to design arrays with increased DOFs
by combining different forms of FODCA while accounting for mutual coupling.
A novel fourth-order hierarchical array (FOHA) based on different forms of FODCA
is constructed using an arbitrary generator set.
The analytical expression between the coupling leakage of the generator and the resulting FOHA is derived.
Two specific FOHA configurations are presented with closed-form sensor placements.
The arrays not only offer increased DOFs for resolving more sources in direction-of-arrival (DOA) estimation
but also effectively suppress mutual coupling.
Moreover, the redundancy of FODCA is examined,
and it is shown that arrays based on the proposed scheme
achieve lower redundancy compared to existing arrays based on FODCA.
Meanwhile, the necessary and sufficient conditions for signal reconstruction by FOHA are derived.
Compared with existing arrays based on FODCA, the proposed arrays provide enhanced DOFs
and improved robustness against mutual coupling.
Numerical simulations verify that FOHAs achieve superior DOA estimation performance compared with other sparse linear arrays.
\end{abstract}

\begin{IEEEkeywords}
    Sparse linear array, fourth-order cumulant, mutual coupling, redundancy, direction of arrival estimation.
\end{IEEEkeywords}

\section{Introduction}
\IEEEPARstart{L}{ow}-cost sampling in intelligent perception systems has been widely applied in various fields,
such as frequency estimation in the time domain \cite{Xiao2017notes},\cite{Xiao2018robustness},\cite{Xiao2016symmetric},\cite{Xiao2021wrapped},\cite{Xiao2023on}
and direction-of-arrival (DOA) estimation in the spatial domain.
The paper primarily focuses on DOA estimation in array signal processing,
a fundamental problem that has been extensively studied for several decades \cite{Krim1996},\cite{Godara1997},\cite{Tuncer2009},\cite{Xiao22023}.
It is well known that a uniform linear array (ULA) with $N$ sensors can estimate at most
$N-1$ sources using classical methods such as MUSIC \cite{Schmidt1986} and ESPRIT \cite{Roy1989}.
To increase the degrees of freedom (DOFs) of a ULA, more sensors must be deployed, which significantly increases hardware cost and complexity.
Moreover, ULAs are susceptible to mutual coupling effects due to the close spacing between adjacent sensors.
In contrast, nonuniform linear arrays, commonly referred to as sparse linear arrays (SLAs), offer an effective solution to these limitations.
For a SLA with $N$ physical sensors, the corresponding difference co-array (DCA) can be formed using the second-order cumulants of the received signals.
The DCA can yield up to $\mathcal{O}(N^2)$ consecutive lags \cite{Pal2010},\cite{Pal2011}, thereby substantially enhancing the DOFs compared to traditional ULAs.
Furthermore, the increased inter-sensor spacing in SLAs can mitigate mutual coupling effects \cite{BD2017mutual}.

The minimum redundancy array (MRA) \cite{Moffet1968} is a foundational structure designed to maximize the DOFs by minimizing sensor redundancy.
However, its lack of closed-form expressions for sensor locations limits scalability and complicates the design of large-scale arrays.
This drawback has motivated extensive research on nested arrays \cite{Pal2010} and coprime arrays \cite{Pal2011},
both of which offer significant advantages by providing closed-form sensor configurations.
The success of nested and coprime arrays has spurred further developments aimed at enhancing DOFs, such as augmented coprime arrays \cite{Pal22011},
enhanced nested arrays \cite{Zhao2019} and arrays designed based on the maximum element spacing criterion \cite{Zheng2019}.
In the context of DOA estimation,
traditional subspace-based methods \cite{Liu2015} rely solely on the consecutive lags of the DCA,
rendering hole-free DCA configurations particularly beneficial.
As a result, hole-filling strategies have been proposed to construct new coprime-like arrays with hole-free DCAs \cite{Wang2019}.
Clearly, various DCA-based array configurations leveraging second-order cumulants have been extensively studied for DOA estimation
due to their ability to significantly enhance the achievable DOFs \cite{Pal2010},\cite{Pal2011},\cite{Zhao2019},\cite{Zheng2019},\cite{Shi2022}.

Previous studies based on second-order statistics lead to the construction of DCAs.
In contrast, high-order cumulant methods, particularly those using fourth-order cumulants (FOC),
have been widely explored for DOA estimation \cite{Xiao2023}, \cite{Guo2024},\cite{Shen2016},\cite{Nikias1993},\cite{Dongan1995},\cite{Chevalier2005},\cite{Shen2019},\cite{Cohen2019}.
Compared to methods based on second-order cumulant,
methods based on FOC provide $\mathcal{O}(N^4)$ virtual sensors with $N$ physical sensors that is much higher than the method based on second-order DCA,
larger virtual apertures, higher resolution, improved accuracy and greater robustness to Gaussian noise \cite{Piya2012},\cite{Dongan1995},\cite{Gonen1999}.
Although they involve higher computational complexity and are ineffective for Gaussian sources, many real-world signals are non-Gaussian \cite{Nikias1993},
making methods based FOC particularly appealing in such scenarios.

For example, the four-level nested array in \cite{Piya2012} is constructed based on the fourth-order difference co-array (FODCA)
of the form $(a - b) - (c - d)$, which generalizes the two-level nested structure originally introduced in \cite{Shen2016}.
An enhanced and simplified version of the four-level nested array is later proposed in \cite{Shen2019},
also based on the FODCA of form $(a - b) - (c - d)$, resulting in increased DOFs.
In \cite{Guo2024}, a new fourth-order sparse array, termed the sum-difference-FODC,
is proposed by incorporating the second-order SCA and DCA into the FODCA of form $(a-b)-(c-d)$.
This design consists of a sparse subarray and another subarray with extended sensor spacing.
By jointly exploiting the second-order SCA and DCA, the resulting FODCA achieves a higher number of consecutive lags.
To further enhance the DOFs, \cite{CaiJJ2017} introduces an expansion and shift scheme
based on the alternative FODCA of form $(a - b) + (c - d)$.
In addition, \cite{ChenH2025} presents a sparse array from a sum-difference co-array perspective,
also based on the form $(a - b) - (c - d)$, to construct arrays with large DOFs.


However, there are limitations remain in the existing studies on SLAs design based on FODCA, as outlined below.

\emph{1)}
Most existing sparse arrays are designed based on a single expression of the FODCA,
either $(a + b) - (c + d)$ or $(a - b) + (c - d)$, derived from circular and symmetric fourth-order cumulant.
This constraint limits the flexibility in selecting sensor position combinations to form FODCA elements.
Consequently, the maximum achievable DOFs of sparse arrays based on FODCA are inherently restricted and cannot be further improved under the current design paradigm.
Whether jointly incorporating both FODCA forms, $(a + b) - (c + d)$ and $(a - b) + (c - d)$, can effectively expand the aperture of the virtual co-array remains an open question. This issue has not been fully addressed in the existing literature
\cite{Xiao2023}, \cite{Guo2024}, \cite{Shen2016}, \cite{Cohen2020}, \cite{Zhou2020}, \cite{DengJ2020},\cite{LiuJHZ2007},\cite{WangB2020}.

\emph{2)}
Most studies on sparse arrays based on FODCA primarily aim to enhance the DOFs,
while overlooking the impact of mutual coupling effects between physical sensors \cite{Xiao2023}, \cite{Guo2024}, \cite{Shen2016}, \cite{Shen2019}, \cite{Cohen2019}, \cite{ChenYP2023}.
However, in practical scenarios, mutual coupling among physical sensors can significantly degrade the DOA estimation performance.

To significantly increase the DOFs of an array with a fixed number of physical sensors based on FOC,
this paper proposes a novel expansion scheme for designing a fourth-order hierarchical array
based on two different forms of FODCA (FOHA).
The proposed scheme leverages two distinct forms of the FODCA, both derived from the received signal data.
Since the effectiveness of FODCA depends on non-zero FOC, the underlying signal is assumed to be non-Gaussian \cite{Guo2024}, \cite{Shen2016}, \cite{Shen2017}, \cite{Zhou2020}.
Under above assumption, the proposed design explores the joint use of both forms of FODCA to construct sparse arrays,
with the goals of improving resolution, extending the virtual aperture and DOFs, and enhancing robustness against mutual coupling and noise.
Furthermore, two specific sparse linear arrays, FOHA(NA) and FOHA(CNA), are developed.
These configurations feature hole-free virtual co-arrays and admit closed-form sensor location expressions,
making them suitable for practical implementation and analytical evaluation.


\emph{Contribution:}
This paper focuses on the design of SLAs for non-Gaussian signals,
aiming to construct hole-free FODCA with enhanced DOFs and reduced mutual coupling.
The main contributions of the paper are summarized as follows:

$\bullet$
A novel scheme is proposed to significantly enhance the DOFs
by jointly utilizing two different forms of FODCA.
The novel scheme enables the construction of FOHA based on any arbitrary generator,
offering broad applicability and flexibility in array design.

$\bullet$
Based on the proposed scheme, two novel arrays, FOHA(NA) and FOHA(CNA),
are constructed using nested array (NA) and concatenated nested array (CNA) as generators.
Each consists of three linear subarrays arranged side-by-side, with sensor positions given by closed-form formulas.
The maximum DOFs are analytically derived by optimizing sensor placement across the subarrays.
These configurations achieve significantly higher DOFs than existing SLAs based on FODCA \cite{Guo2024}, \cite{Piya2012}, \cite{Shen2019}, \cite{Yang2023}.

$\bullet$
A necessary and sufficient condition for signal reconstruction is the existence of a positive integer $k$.
This integer must be at least twice the ratio of the least common multiple of all sensor positions to the wavelength.
Furthermore, the necessary and sufficient condition for signal reconstruction based on FOHA is also derived.

$\bullet$
Moreover, the redundancy expressions for the proposed FOHA(NA) and FOHA(CNA) are analytically derived,
offering further insights into their design advantages in terms of physical sensor economy.
In addition, the expression between the coupling leakage of the generator array
and the resulting FOHA is explicitly derived.
Leveraging the scaling factors $\lambda_1$ and $\lambda_2$, the generated arrays exhibit increased sparsity,
which reduces coupling leakage and enhances robustness against mutual coupling effects.

This paper is organized as follows.
Section II reviews the general sparse array model and mutual coupling model,
and introduces the formulation of FODCA.
Section III presents a novel scheme for constructing FOHA based on a generator array,
along with the conditions for achieving a hole-free FODCA.
In Section IV, two specific FOHA configurations, FOHA(NA) and FOHA(CNA),
are proposed using NA and CNA as generators, respectively.
Their hole-free properties and maximum achievable DOFs are analyzed.
Section V derives the necessary and sufficient conditions for signal reconstruction based on FOHA.
Sections VI and VII analyze the lower bounds of fourth-order redundancy and mutual coupling leakage in FOHA arrays, respectively.
Section VIII demonstrates the advantages of FOHA(NA) and FOHA(CNA) in terms of DOFs enhancement and improved robustness to mutual coupling.
Their DOFs and coupling leakage performance are compared with six existing FODCA configurations.
Numerical simulations are also conducted to evaluate RMSE performance under varying SNR levels, snapshot numbers and source counts, both with and without mutual coupling.

\textit{Notations:}
$\mathbb{P}$ is the physical sensor positions set of a SLA.
$N$ is the number of sensors. $D$ is the number of source signals to be estimated. $K$ is the number of snapshots.
$\Phi^u$ is sensor positions set of a FODCA.
The operators $\otimes$, $\odot$, $(\cdot)^T$, $(\cdot)^H$ and $(\cdot)^*$ stand for the Kronecker products,
Khatri-Rao products,
transpose, conjugate transpose and complex conjugation, respectively.
Set $\{a : b : c \}$ denotes the integer line from $a$ to $c$ sampled in steps of $b\in \mathbb{N}^+$.
When $b=1$, we use shorthand $\{a : c \}$ .

\section{Preliminaries}

\subsection{General Sparse Array Model}
Assume that there are $D$ non-Gaussian and mutually uncorrelated far-field narrow band signals.
The incident angle of the $i^{th}$ signal is $\theta_i$,
and the physical sensor positions set of the SLA is represented as $\mathbb{P}d=\{p_k|k\in[1,N]\}d$,
where $\mathbb{P}=\{ p_1,p_2,...,p_N \}$ is the set of sensor position normalized by the unit spacing $d$,
which is generally set to half wavelength.
The array output at the $n^{th}$ physical sensor corresponding to the $t^{th}$ snapshot,
denoted as $x_n(t)$, can be expressed as follows
\begin{equation}
\label{w1}
\begin{aligned}
x_n(t)=\sum_{i=1}^Da_n(\theta_i)s_i(t)+v_n(t),
\end{aligned}
\end{equation}
where $a_n(\theta_i)$ denotes the steering response of $n^{th}$ physical sensor corresponding to the $i^{th}$ source signal,
which can be expressed as follows
\begin{equation}
\label{w8}
a_n(\theta_i)=e^{j\frac{2\pi p_{n} d}{\lambda}sin(\theta_i)}.
\end{equation}

Further, $v_n(t)$ denotes a zero-mean additive Gaussian noise sample at the $n^{th}$ physical sensor, which is assumed to be
statistically independent of all the source signals. Thus, the received signals for all the $N$ physical sensors, which are denoted as
$\boldsymbol{x}(t)=[x_1(t),...,x_N(t)]^T$, can be written as follows
\begin{equation}
\label{w2}
\boldsymbol{x}(t)=\sum_{i=1}^D\boldsymbol{a}(\theta_i)s_i(t)+\boldsymbol{v}(t)=\boldsymbol{A}(\theta)\boldsymbol{s}(t)+\boldsymbol{v}(t),
\end{equation}
where $\boldsymbol{s}(t) = [s_1(t), \dots, s_D(t)]^T$ denotes the vector of non-Gaussian source signals,
$\boldsymbol{v}(t) = [v_1(t), \dots, v_N(t)]^T$ denotes the vector of additive Gaussian noise,
and $\boldsymbol{a}(\theta_i) = [a_1(\theta_i), \dots, a_N(\theta_i)]^T$ denotes the array steering vector corresponding to the $i^{\text{th}}$ source signal.
The array manifold matrix is defined as $\boldsymbol{A}(\theta) = [\boldsymbol{a}(\theta_1), \dots, \boldsymbol{a}(\theta_D)] \in \mathbb{C}^{N \times D}$.
By setting $p_1 = 0$, the steering vector $\boldsymbol{a}(\theta_i)$ can be rewritten as
\begin{equation}\nonumber
\boldsymbol{a}(\theta_i)=[1,e^{-j\frac{2\pi p_2d\sin\theta_i}{\lambda}},...,e^{-j\frac{2\pi p_Nd\sin\theta_i}{\lambda}}]^T.
\end{equation}

\subsection{Mutual Coupling Model}
There exists the mutual coupling effects among the physical sensors in practical applications.
When considering the effect of mutual coupling,
the received signal vector in (\ref{w2}) can be rewritten as
\begin{equation}
\label{wang5}
\boldsymbol{x}(t)=\boldsymbol{C}\boldsymbol{A}(\theta)\boldsymbol{s}(t)+\boldsymbol{v}(t),
\end{equation}
where $\boldsymbol{C}$ is the $N \times N$ mutual coupling matrix.
In general, the expression for $\boldsymbol{C}$ is typically quite complicated~\cite{LiuJ2017}.
However, under the ULA configuration, $\boldsymbol{C}$ can be approximated by a
$B$-banded symmetric Toeplitz matrix, as shown below~\cite{Friedlander1991},~\cite{Svantesson19991},~\cite{Svantesson19992}.
\begin{equation}
\boldsymbol{C}(p_{k_1},p_{k_2})=
\begin{cases}
c_{|p_{k_1}-p_{k_2}|},\ \ \ \ if\ |p_{k_1}-p_{k_2}|\leq B,\\
0,\ \ \ \ \ \ \ \ otherwise,
\end{cases}
\end{equation}
where $p_{k_1},p_{k_2}\in\mathbb{P}$, and $c_0,c_1,...,c_B$ are coupling coefficients satisfying $c_0=1>|c_1|>|c_2|>...>|c_B|$ \cite{Friedlander1991}.
For a given array with $N$ sensors, the coupling leakage is defined as the energy ratio \cite{Zheng2019}
\begin{equation}
\label{w24}
L=\frac{\|\boldsymbol{C}-diag(\boldsymbol{C}) \|_F}{\|\boldsymbol{C} \|_F},
\end{equation}
a small value of $L$ implies that the mutual coupling is less significant.

\subsection{The Fourth-Order Cumulant and Virtual Array}
Next, we analyze the fourth-order circular cumulant of the received signals indexed by $k_1, k_2, k_3, k_4 \in [1, N]$.
According to  \cite{DengJ2020}, \cite{Yang2023}, the fourth-order circular cumulant of $x_{k_1}(t)$, $x_{k_2}(t)$, $x_{k_3}(t)$ and $x_{k_4}(t)$ is can be expressed as following two forms
\begin{equation}
\label{eq1}
\begin{aligned}
&\text{cum}(x_{k_1} (t), x_{k_2}(t), x_{k_3}^*(t), x_{k_4}^*(t))\\
&:=\mathbb{E}[x_{k_1}\sx(t)x_{k_2}\sx(t)x_{k_3}^*\sx(t)x_{k_4}^*\sx(t)]
-\mathbb{E}[x_{k_1}\sx(t)x_{k_2}\sx(t)]\mathbb{E}[x_{k_3}^*\sx(t)x_{k_4}^*\sx(t)]\\
&-\mathbb{E}[x_{k_1}\sx(t)x_{k_3}^*\sx(t)]\mathbb{E}[x_{k_2}\sx(t)x_{k_4}^*\sx(t)]
-\mathbb{E}[x_{k_1}\sx(t)x_{k_4}^*\sx(t)]\mathbb{E}[x_{k_2}\sx(t)x_{k_3}^*\sx(t)]\\
&\text{cum}(x_{k_1}(t), x_{k_2}^*(t), x_{k_3}(t), x_{k_4}^*(t))\\
&:=\mathbb{E}[x_{k_1}\sx(t)x_{k_2}^*\sx(t)x_{k_3}\sx(t)x_{k_4}^*\sx(t)]
-\mathbb{E}[x_{k_1}\sx(t)x_{k_2}^*\sx(t)]\mathbb{E}[x_{k_3}\sx(t)x_{k_4}^*\sx(t)]\\
&-\mathbb{E}[x_{k_1}\sx(t)x_{k_3}\sx(t)]\mathbb{E}[x_{k_2}^*\sx(t)x_{k_4}^*\sx(t)]
-\mathbb{E}[x_{k_1}\sx(t)x_{k_4}^*\sx(t)]\mathbb{E}[x_{k_2}^*\sx(t)x_{k_3}\sx(t)].
\end{aligned}
\end{equation}

According to Theorem 1 in \cite{Piya2012} and assuming independence between the source signals and the noise,
the fourth-order cumulant in (\ref{eq1}) simplifies to
\begin{equation}
\label{eq2}
\begin{aligned}
&\text{cum}(x_{k_1}(t), x_{k_2}(t), x_{k_3}^*(t), x_{k_4}^*(t))\\
&=\sum_{i=1}^De^{-j\frac{2\pi d(p_{k_1}+p_{k_2}-p_{k_3}-p_{k_4})}{\lambda}\sin(\theta_i)}c_{4,s_i}^1,\\
&\text{cum}(x_{k_1}(t), x_{k_2}^*(t), x_{k_3}(t), x_{k_4}^*(t))\\
&=\sum_{i=1}^De^{-j\frac{2\pi d(p_{k_1}-p_{k_2}+p_{k_3}-p_{k_4})}{\lambda}\sin(\theta_i)}c_{4,s_i}^2,
\end{aligned}
\end{equation}
where $c_{4,s_i}^1 := \text{cum}(s_i(t), s_i(t), s_i^*(t), s_i^*(t))$ and
$c_{4,s_i}^2 := \text{cum}(s_i(t), s_i^*(t), s_i(t), s_i^*(t))$ denote the fourth-order cumulant of the $i$th source.
Note that the expression in (\ref{eq2}) resembles the array model in (\ref{eq1}).
The term $p_{k_1} + p_{k_2} - p_{k_3} - p_{k_4}$ and $p_{k_1} - p_{k_2} + p_{k_3} - p_{k_4}$
in (\ref{eq2}) can be interpreted as the position of a sensor in FODCA.
Motivated by this observation, we define the FODCA as follows, in accordance with Definition~1 in \cite{Piya2012}.
\begin{definition}
For a linear array $\mathbb{P}=\{p_1, p_2,..., p_N\}$,
its FODCA is defined as
\begin{equation}
\label{eq3}
\begin{aligned}
\Delta_4(\mathbb{P})\sx &:=\sx \{(p_{k_1} +p_{k_2})-(p_{k_3}+p_{k_4}) | k_1,k_2,k_3,k_4\in [1,N]\}\\
\sx &\cup \{(p_{k_1} -p_{k_2})+(p_{k_3}-p_{k_4}) |  k_1,k_2,k_3,k_4\in [1,N]\}.
\end{aligned}
\end{equation}
\end{definition}


\section{A Novel Design Scheme for Fourth-Order Hierarchical Array}
\subsection{ The Novel Scheme for FOHA}
The purpose of designing sparse array structures is to enhance DOFs.
Firstly, criteria in sparse array designing are described.

\emph{Criterion 1 (Large consecutive lags of DCA)}:
The large consecutive lags of DCA are preferred, which can not only increase the number of resolvable sources
but also lead to higher spatial resolution in DOA estimation \cite{Pal2010}, \cite{Piya2012}, \cite{Shen2019}.

\emph{Criterion 2 (Closed-form expressions of sensor positions)}:
A closed-form expression of sensor positions is preferred for scalability considerations \cite{Cohen2019}, \cite{Cohen2020}.

\emph{Criterion 3 (Hole-free DCA)}:
A SLA with a hole-free DCA is preferred,
since the data from its DCA can be utilized directly by subspace-based DOA estimation methods
which are easy to be implemented, and thus the algorithms based on compressive sensing \cite{Shen2017}
or co-array interpolation techniques \cite{Cui2019}, \cite{Zhou2018} with increased computational complexity can be avoided \cite{Cohen2020}, \cite{Liu20172}.

Secondly, consider the case where $\mathbb{P}$ is decomposed into three disjoint subsets as
$\mathbb{P} = \mathbb{A}_1 \cup \mathbb{A}_2 \cup \mathbb{A}_3$.
Since the elements in $\Delta_4(\mathbb{P})$ are symmetric with respect to the origin,
only the non-negative part of $\Delta_4(\mathbb{P})$ is considered in the following analysis.
Furthermore, the FODCA, $\Delta_4(\mathbb{P})$, is closely related to second-order and third-order co-arrays.
The second-order and third-order DCA of an array $\mathbb{P}$ are defined as follows
\begin{equation}
\begin{aligned}\nonumber
&\Delta_2(\mathbb{P}):=\mathbb{P}-\mathbb{P}=\{ m-n\ |\ m,n \in \mathbb{P}  \},\\
&\Delta_{31}(\mathbb{P}):=\mathbb{P}+\mathbb{P}-\mathbb{P}=\{ m+n-z\ |\ m,n,z \in \mathbb{P}  \},\\
&\Delta_{32}(\mathbb{P}):=\mathbb{P}-\mathbb{P}-\mathbb{P}=\{ m-n-z\ |\ m,n,z \in \mathbb{P}  \}.
\end{aligned}
\end{equation}

The second-order SCA of an array $\mathbb{P}$ is defined as
\[
\Sigma_2(\mathbb{P}):=\mathbb{P}+\mathbb{P}=\{ m+n\ |\ m,n \in \mathbb{P}  \}.
\]

To improve the continuity of $\Delta_4(\mathbb{P})$, we first aim to maximize the cardinalities of the second-order co-arrays.
\begin{equation}
\begin{aligned}
\label{eq:delta2_sigma2}
& \Delta_2(\mathbb{A}_1) := \{p_{k_1} - p_{k_2} \mid p_{k_1}, p_{k_2} \in \mathbb{A}_1 \}, \\
&\Sigma_2(\mathbb{A}_1) := \{p_{k_1} + p_{k_2} \mid p_{k_1}, p_{k_2} \in \mathbb{A}_1 \}.
\end{aligned}
\end{equation}

Next, for each $p_{k_3} \in \mathbb{P}_2$, we consider the sets
\begin{equation}\nonumber
\begin{aligned}
\Delta_{31}(\mathbb{A}_1, \mathbb{A}_2) &:= \{p_{k_3} - a \mid a \in \Delta_2,\ p_{k_3} \in \mathbb{A}_2 \}, \\
\Delta_{32}(\mathbb{A}_1, \mathbb{A}_2) &:= \{p_{k_3} - b \mid b \in \Sigma_2,\ p_{k_3} \in \mathbb{A}_2 \},
\end{aligned}
\end{equation}
which can extend the non-negative support of $\Delta_4(\mathbb{P})$ to the interval $[0,\ p_{k_3}]$.

For each element \( p_{k_4} \in \mathbb{A}_3 \), additional coverage in the non-negative region is obtained by
\begin{equation}\nonumber
\begin{aligned}
\Delta_4(\mathbb{P}) &:= \left\{p_{k_4} - e\ | \ e \in \Delta_{31},\ p_{k_4} \in \mathbb{A}_3 \right\} \\
 &\cup  \left\{p_{k_4} + f\ | \ f \in \Delta_{32},\ p_{k_4} \in \mathbb{A}_3 \right\}.
\end{aligned}
\end{equation}

Similarly, for each \( p_{k_4} \in \mathbb{A}_3 \), coverage in the negative region is given by
\begin{equation}\nonumber
\begin{aligned}
\Delta_4(\mathbb{P})&:= \left\{-\left[p_{k_4} - e \right]\right\}=\left\{-p_{k_4} + p_{k_3} + p_{k_1} - p_{k_2} \right\}\\
&\cup \left\{-\left[p_{k_4} + f \right]\right\}= \left\{-p_{k_4} - p_{k_3} + p_{k_1} + p_{k_2} \right \},
\end{aligned}
\end{equation}
ensuring symmetry in \(\Delta_4(\mathbb{P})\); that is, each value appears together with its negative counterpart, maintaining the structural balance of \(\Delta_4(\mathbb{P})\).

In addition, to lighten the notations, given any three sets $\mathbb{P}$, $\mathbb{P}'$ and $\mathbb{P}''$ \cite{Xiao2023}, we use
\begin{equation}\nonumber
\begin{aligned}
&\mathbb{C}(\mathbb{P},\mathbb{P}')=\{ p_{k_1}+p_{k_2}\ | \ p_{k_1}\in \mathbb{P},p_{k_2}\in \mathbb{P}' \},\\
&\mathbb{C}(\mathbb{P},\mathbb{P}',\mathbb{P}'')\sx=\sx \{ p_{k_1}\sx+\sx p_{k_2}\sx+\sx p_{k_3}\ | \ p_{k_1}\in \mathbb{P},p_{k_2}\in \mathbb{P}',p_{k_3}\in \mathbb{P}'' \},
\end{aligned}
\end{equation}
to denote the cross sum of elements from $\mathbb{P}$, $\mathbb{P}'$ and $\mathbb{P}''$.

For convenience, the largest element in the set \(\mathbb{P} \) is denoted as $\max(\mathbb{P})$.

\begin{remark}
Each stage of the construction, \(\Delta_{31}\), \(\Delta_{32}\) and \(\Delta_{4}\),
builds upon elements generated in previous stages.
Specifically, \(\mathbb{A}_2\) utilizes the structure of \(\mathbb{A}_1\) to extend coverage over the range \([0, \max(\mathbb{A}_2)]\),
while \(\mathbb{A}_3\) reverses and shifts these elements to cover the negative axis.
\end{remark}

Based on the above analysis, the proposed scheme to construct a hole-free FODCA consists of four main steps:
\begin{enumerate}
    \item Partition the original set \(\mathbb{P}\) into three disjoint subsets: \(\mathbb{A}_1\), \(\mathbb{A}_2\) and \(\mathbb{A}_3\).

    \item Use the first subset \(\mathbb{A}_1\) as the generator to construct the second-order SCA \(\Sigma_2(\mathbb{A}_1)\) and the second-order DCA \(\Delta_2(\mathbb{A}_1)\).

    \item Generate the second subset \(\mathbb{A}_2\) by extending \(\Sigma_2(\mathbb{A}_1)\). Based on \(\mathbb{A}_1 \cup \mathbb{A}_2\), derive two distinct third-order DCAs: \(\Delta_{31}(\mathbb{A}_1 , \mathbb{A}_2)\) and \(\Delta_{32}(\mathbb{A}_1, \mathbb{A}_2)\).

    \item Obtain the third subset \(\mathbb{A}_3\) by further extending \(\Delta_{31}(\mathbb{A}_1 , \mathbb{A}_2)\) and \(\Delta_{32}(\mathbb{A}_1, \mathbb{A}_2)\).
\end{enumerate}

Therefore, the FODCA of $\mathbb{P}$ can be obtained.
Moreover, its consecutive segments can be expressed as \(\mathbb{U}(\mathbb{P}) = \{-U, \ldots, U\}\),
where \(U \in \mathbb{Z}_{\geq 0}\) is chosen to be maximal.
The design goal is to maximize the size of this segment, \(|\mathbb{U}(\mathbb{P})| = 2U + 1\),
subject to structural constraints on \(\mathbb{A}_1\), \(\mathbb{A}_2\) and \(\mathbb{A}_3\).

\subsection{FOHA Structure}
This section proposes a novel scheme for designing a fourth-order hierarchical array based on distinct forms of FODCA (FOHA).
The generator acts as the fundamental building block for
constructing a sparse array and can be selected from any desired array configuration.
\begin{definition}
(FOHA with generator $\mathbb{A}_1$): The physical sensor positions of the FOHA are set as

\begin{equation}\nonumber
\begin{aligned}
&\mathbb{P}= \mathbb{A}_1\cup \mathbb{A}_2\cup \mathbb{A}_3,\\
&\mathbb{A}_2=\{ \delta_1:\eta_1: \delta_1+\eta_1(N_2-1) \},\\
&\mathbb{A}_3=\{ \delta_2:\eta_2: \delta_2+\eta_2(N_3-1) \},
&\end{aligned}
\end{equation}
where
\begin{equation}\nonumber
\begin{aligned}
&\delta_1=\max(\mathbb{A}_1)+\lambda_1+1,\ \eta_1=\lambda_1+1, \\
&\delta_2\sx=\sx (\frac{1}{2}\sx+\sx N_2)\lambda_1\sx+\sx \lambda_3\sx+\sx \lambda_4\sx+\sx N_2\sx+\sx 1,\
\eta_2\sx=\sx \lambda_3\sx+\sx \lambda_4\sx+\sx 1,
\end{aligned}
\end{equation}
and the values of $\lambda_i(i=1,2,3,4)$ are
\begin{equation}\nonumber
\begin{aligned}
&\mathbb{U}_1=\{a+b\ |\ a,b \in \mathbb{A}_1\}=\{0: \lambda_1\},\\
&\mathbb{U}_2=\{a-b\ |\ a,b \in \mathbb{A}_2\}=\{ -\lambda_2:\lambda_2\}\ \& \ \lambda_1\geq\lambda_2, \\
&\mathbb{U}_3=\{c-(a+b)\ |\ a,b \in \mathbb{A}_1, c\in\mathbb{A}_2 \}=\{0: \lambda_3 \},\\
&\mathbb{U}_4=\{c-(a-b)\ |\ a,b \in \mathbb{A}_2, c\in\mathbb{A}_2 \}=\{0:\lambda_4 \}.
\end{aligned}
\end{equation}
\end{definition}

\begin{figure*}
 \center{\includegraphics[width=14cm]  {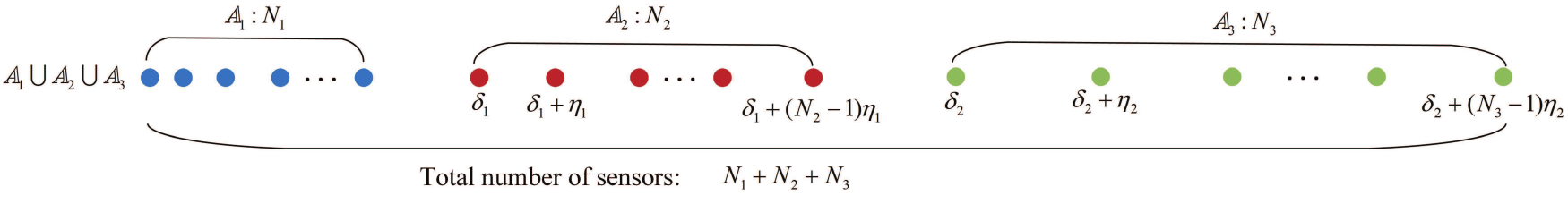}}
 \caption{\label{1} Structure of FOHA}
 \label{fig:foc}
\end{figure*}

\subsection{Consecutive segment of FOHA}
\begin{lemma}
\label{lea:dof}
FOHA with generator \(\mathbb{A}_1\) designed based on FODCA
is hole-free if the following conditions hold:\\
(I) $\lambda_1\geq\max(\mathbb{A}_1)$.\\
(II) $\lambda_2\geq\max(\mathbb{A}_1)$.\\
(III) $\lambda_2\geq \frac{\lambda_1}{2}$.
\end{lemma}
\begin{proof}
Firstly, since the consecutive segment of the FODCA for FOHA are symmetric about zero, only the positive lags are considered.
The second-order DCA and SCA of the first subarray \(\mathbb{A}_1\) are given by
\begin{equation}\nonumber
\begin{aligned}
&\Sigma_2(\mathbb{A}_1)=\mathbb{C}(\mathbb{A}_1,\mathbb{A}_1)\supseteq\mathbb{U}_1=\{ 0 : \lambda_1\} , \\
&\Delta_2(\mathbb{A}_1)=\mathbb{C}(\mathbb{A}_1,-\mathbb{A}_1)\supseteq\mathbb{U}_2=\{ -\lambda_2 : \lambda_2 \}.
\end{aligned}
\end{equation}

Based on the obtained second-order DCA and SCA, two distinct third-order DCAs of \(\mathbb{A}_1 \cup \mathbb{A}_2\) can be derived as follows
\begin{equation}\nonumber
\begin{aligned}
&\Delta_{31}(\mathbb{A}_2,-\Sigma_2(\mathbb{A}_1))=\mathbb{C}(\mathbb{A}_2,-\mathbb{C}(\mathbb{A}_1,\mathbb{A}_1))\supseteq\\
&\mathbb{U}_3\sx=\sx \{\max(\mathbb{A}_1) + 1 : N_2\lambda_1 + \max(\mathbb{A}_1) + N_2\}, \\
&\Delta_{31}(\max(\mathbb{A}_1),-\Sigma_2(\mathbb{A}_1))\sx=\sx \mathbb{C}(\max(\mathbb{A}_1),-\mathbb{C}(\mathbb{A}_1,\mathbb{A}_1))\supseteq\\
&\mathbb{U}_4=\{ \max(\mathbb{A}_1)-\lambda_1 : \max(\mathbb{A}_1) \}, \\
&\Delta_{32}(\mathbb{A}_2,-\Delta_2(\mathbb{A}_1))=\mathbb{C}(\mathbb{A}_2,-\mathbb{C}(\mathbb{A}_1,-\mathbb{A}_1))\supseteq\\
&\mathbb{U}_5\sx=\sx\{ \max(\mathbb{A}_1)\sx+\sx\lambda_1\sx+\sx 1\sx-\sx\lambda_2 \sx : \sx N_2\lambda_1\sx+\sx \max(\mathbb{A}_1)\sx+\sx N_2\sx+\sx \lambda_2 \}, \\
&\Delta_{32}(\max(\mathbb{A}_1),-\Delta_2(\mathbb{A}_1))=\mathbb{C}(\lambda_2,-\mathbb{C}(\mathbb{A}_1,-\mathbb{A}_1))\supseteq\\
&\mathbb{U}_6=\{ \max(\mathbb{A}_1)-\lambda_2 : \max(\mathbb{A}_1)+\lambda_2  \}.
\end{aligned}
\end{equation}

When condition (I) holds, it follows that \(\max(\mathbb{A}_1) - \lambda_1 \leq 0\).
Therefore, the following equation can be established
\begin{equation}\nonumber
\begin{aligned}
&\mathbb{U}_3 \cup \mathbb{U}_4 \supseteq \mathbb{U}_{7}=\{0 : N_2\lambda_1+ \max(\mathbb{A}_1 )+ N_2 \}.
\end{aligned}
\end{equation}

Moreover, when conditions (II) and (III) hold, we have \(\max(\mathbb{A}_1) - \lambda_2 \leq 0\)
and $\max(\mathbb{A}_1) + \lambda_2 + 1 \geq \max(\mathbb{A}_1) + 1$.
Therefore, the following equations hold
\begin{equation}\nonumber
\begin{aligned}
&\mathbb{U}_5 \cup \mathbb{U}_6 \supseteq \mathbb{U}_8=\{0 : N_2\lambda_1+ \max(\mathbb{A}_1)+ N_2+ \lambda_2\}.
\end{aligned}
\end{equation}

Since the FODCA contains two positive and two negative components, the expression for the FODCA can be derived as follows
\begin{equation}\nonumber
\begin{aligned}
&\Delta_4(\max(\mathbb{U}_7),\mathbb{U}_7)\supseteq\mathbb{U}_9\\
&=\{N_2\lambda_1+ \max(\mathbb{A}_1)+N_2 : 2N_2\lambda_1+ \max(\mathbb{A}_1)+N_2 \},\\
&\Delta_4(\mathbb{P}_3,\mathbb{U}_7)\supseteq\mathbb{U}_{10}=\{\delta_2 : \delta_2+N_2\lambda_1+ \max(\mathbb{A}_1)+N_2 \}\\
&\cup \{\delta_2+\eta_2 : \delta_2+\eta_2+N_2\lambda_1+ \max(\mathbb{A}_1)+N_2 \}\\
&\cup \cdots \cdots \cdots\\
&\cup\{\delta_2\sx+\sx \eta_2(N_3\sx-\sx 1) : \delta_2\sx+\sx \eta_2(N_3\sx-\sx 1)\sx+\sx N_2\lambda_1\sx+\sx \max(\mathbb{A}_1)\sx+\sx N_2 \},\\
&\Delta_4(\mathbb{P}_3,\mathbb{U}_8)\supseteq\mathbb{U}_{11}\\
&=\{\delta_2-N_2\lambda_1- \max(\mathbb{A}_1)-N_2-\lambda_2 : \delta_2 \}\\
&\cup \{\delta_2\sx+\sx \eta_2\sx-\sx N_2\lambda_1\sx-\sx \max(\mathbb{A}_1)\sx-\sx N_2\sx-\sx \lambda_2 : \delta_2\sx+\sx \eta_2 \}\\
&\cup \cdots \cdots \cdots\\
&\cup \{\delta_2\sx+\sx\eta_2(N_3\sx-\sx 1)\sx-\sx N_2\lambda_1\sx-\sx \max(\mathbb{A}_1)\sx-\sx N_2\sx-\sx\lambda_2 : \delta_2\sx+\sx\eta_2(N_3\sx-\sx1) \}.
\end{aligned}
\end{equation}

When condition (II) holds, we have $2\max(\mathbb{A}_1) + 2N_1 + \lambda_1 + 2N_2 + 1 \geq \delta_2 - \max(\mathbb{A}_1) - N_2 \lambda_2 - N_2 - \lambda_2$.
Therefore, \(\mathbb{U}_9 \cup \mathbb{U}_{11}\) is hole-free.
Furthermore, the quantity \(\delta_2 + \max(\mathbb{A}_1) + N_2 \lambda_1 + N_2 + 1\) in \(\mathbb{U}_8\) equals
$\delta_2 + \eta_2 - \max(\mathbb{A}_1) - N_2 \lambda_1 - N_2 - \lambda_2$ in \(\mathbb{U}_{10}\) when \(N_3 = 1\).
Similarly, this relationship holds for all \(N_3\).
Therefore, the FODCA of \(\mathbb{P}\) can be expressed as
\begin{equation}\nonumber
\begin{aligned}
\mathbb{U}_{9}\cup\sx \mathbb{U}_{10}\cup\sx \mathbb{U}_{11}\sx=\sx \{0 : \delta_2\sx+\sx \eta_2(N_3\sx-\sx1)\sx+\sx N_2\lambda_1\sx+\sx \max(\mathbb{A}_1)\sx+\sx N_2 \}.
\end{aligned}
\end{equation}

Furthermore, the maximum number of consecutive segment in the FODCA for FOHA is
$2 \left[ \delta_2 + \eta_2 (N_3 - 1) + N_2 \lambda_1 + \max(\mathbb{A}_1) + N_2 \right] + 1$.
\end{proof}

\section{Two Specific FOHAs With Distinct Generators}
\subsection{Generator With Nested Array}
To enhance the DOFs, the generator is chosen as a nested array (NA).
The FOHA(NA) is systematically designed by appropriately positioning the sensors of three subarrays,
as illustrated in Fig.~\ref{fig:stru},
where subarray 1 is a NA with \(N_1\) sensors, and subarrays 2 and 3 are SLAs with large inter-sensor spacings.
\subsubsection{Structure of the Proposed FOHA(NA)}
\begin{definition}
\label{def:na}
The FOHA(NA) consists of three subarrays with $N \sx=\sx N_1 + N_2 + N_3$ sensors,
where \(N_i\) is the number of sensors in the subarray $i$ \((i=1,2,3)\).
The sensor positions, denoted by \(\mathbb{A}_1\), \(\mathbb{A}_2\) and \(\mathbb{A}_3\), are given as follows
\begin{equation}\nonumber
\label{cn1}
\begin{aligned}
&\mathbb{P}=\mathbb{A}_1\cup\mathbb{A}_2\cup\mathbb{A}_3,\\
&\mathbb{A}_1=\{ 0:M_1-1 \} \sx \cup\sx \{ 2M_1-1:M_1:M_1(M_2+1)-1 \},\\
&\mathbb{A}_2=\{\delta_1: \eta_1: \delta_1+\eta_1(N_2-1)\},\\
&\mathbb{A}_3=\{\delta_2: \eta_2: \delta_2+\eta_2(N_3-1)\},\\
&\delta_1=2M_1M_2+3M_1-2, \eta_1=M_1M_2+2M_1-1,\\
&\delta_2=(3N_2+4)M_1M_2+(6N_2+4)M_1-(3N_2+3),\\
&\eta_2=(2N_2+3)M_1M_2+(4N_2+3)M_1-(2N_2+2).
\end{aligned}
\end{equation}
where \(M_1\) and \(M_2\) denote the numbers of physical sensors in subarray 1 and 2 of \(\mathbb{A}_1\), respectively.
The structure of FOHA(NA) is illustrated in Fig.~\ref{fig:stru}(a).
\end{definition}

\begin{figure*}
  \centering
  \subfigure[]{
    \includegraphics[scale=0.185]{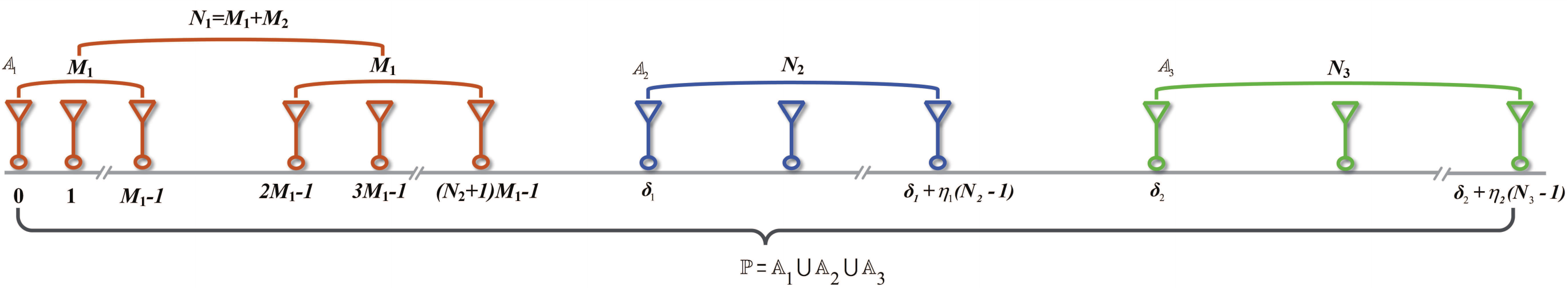}}
  \hspace{0in} 
  \subfigure[]{
    \includegraphics[scale=0.18]{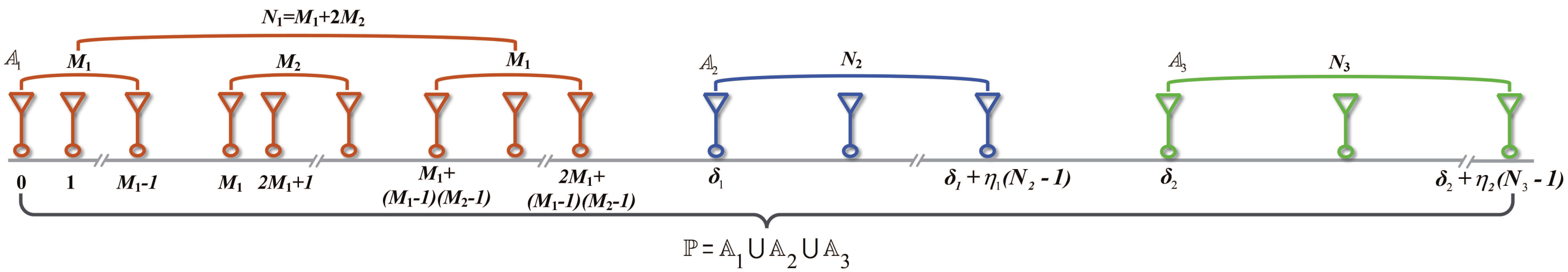}}
 \caption{\label{1} The structure of FOHA with different generators when $N=9$. (a) FOHA(NA). (b) FOHA(CNA). }
\label{fig:stru}
\end{figure*}

\subsubsection{Consecutive Segment of the FOHA(NA)}

\begin{theorem}
FOHA(NA) designed based on the FODCA  is hole-free.
\end{theorem}

\begin{proof}
According to Lemma \ref{lea:dof}, FOHA(NA) designed based on FODCA is hole-free provided that conditions (I), (II) and (III) are satisfied.
Regarding condition (I), we have
\begin{equation}\nonumber
\begin{aligned}
\lambda_2 = M_1 M_2 + M_1 - 1= \max(\mathbb{A}_1),
\end{aligned}
\end{equation}
which satisfies the requirement.

Similarly, for condition (II),
\begin{equation}\nonumber
\begin{aligned}
&\lambda_1=M_1M_2+2M_1-2 > \max(\mathbb{A}_1)=M_1M_2+M_1-1,
\end{aligned}
\end{equation}
ensuring that condition (II) is met.

For condition (III):
\begin{equation}\nonumber
\begin{aligned}
&\lambda_2=M_1M_2+M_1-1 > \frac{\lambda_1}{2}=\frac{1}{2}M_1M_2+M_1-1,
\end{aligned}
\end{equation}
which also holds.

Therefore, by Lemma \ref{lea:dof}, FOHA(NA) designed based on FODCA is hole-free when conditions (I), (II) and (III) are satisfied.
\end{proof}

\subsubsection{The Maximum DOFs of FOHA(NA) With the Given Number of Physical Sensors}
~\par
The DOFs of FOHA(NA) can be further improved by optimizing the allocation of physical sensors among subarrays 1, 2 and 3
for a given number of physical sensors.
\begin{theorem}
To maximize the DOFs of FOHA (NA) with a fixed total number of sensors,
the number of sensors for each subarray is set as
\begin{equation}\nonumber
\begin{cases}
N_1=Caculated\ by\ Algorithm\ 1,\\
N_2=\lceil\frac{-\frac{1}{4}N_1^3+(\frac{N}{4}-\frac{17}{8})N_1^2+(2N+\frac{5}{2})N_1+2N}{ \frac{1}{2}N_1^2 + 4N_1 - 4}\rfloor,\\
N_3=\lceil\frac{-\frac{1}{4}N_1^3+(\frac{N}{4}-\frac{15}{8})N_1^2+(2N+\frac{3}{2})N_1-2N}{ \frac{1}{2}N_1^2 + 4N_1 - 4}\rfloor.
\end{cases}
\end{equation}

For subarray 1, to maximize the DOFs, the parameters \(M_1\) and \(M_2\) are chosen as
\begin{equation}\nonumber
M_1=\lceil N_1/2 \rceil, \ M_2=\lfloor N_1/2 \rfloor.
\end{equation}
\end{theorem}

\begin{proof}
To obtain the maximum DOFs of FOHA(NA),
it is equivalent to solve the following optimization problem
\begin{equation}\nonumber
\begin{aligned}
& \underset{N_1,N_2,N_3\in \mathbb{N}_+}{maximize}\ \ \delta_2 + \eta_2(N_3 - 1) + \lambda_3,\\
& subject\ to\ \ N_1+N_2+N_3=N. \ \ \ \ \ \ \ \ \ \ \ \ \ \ (P1)
\end{aligned}
\end{equation}

For problem (P1), the objective function is cubic.
Since the second derivative (i.e., curvature) of a cubic function changes sign,
the curve can bend both upwards and downwards.
This property renders the optimization problem non-convex.
In a non-convex optimization problem, multiple local extrema may exist within the domain,
making it difficult to directly obtain the global optimum \cite{Krentel1986}.
To address the problem,
by assuming \(N_1\) for subarray 1 is fixed and chosen to maximize the consecutive lags
of the second-order SCA and DCA generated by subarray 1,
problem (P1) is reformulated as an optimization problem with only two variables, \(N_2\) and \(N_3\), expressed as follows
\begin{equation}\nonumber
\begin{aligned}
&\underset{N_1,N_2,N_3\in \mathbb{N}_+}{maximize}\ \ \delta_2 + \eta_2(N_3 - 1) + \lambda_3,\\
&subject\ to\ \ \ N_2+N_3=N-N_1. \ \ \ \ \ \ \ \ \ \ \ \ \ \ \ \ \ \ \ \ (P2)
\end{aligned}
\end{equation}

At this point, problem (P2) becomes a convex optimization problem,
which guarantees that the global optimum can be directly obtained within the domain of the independent variables \cite{Krentel1986}.
To solve problem (P2), substitute $N_2 = N - N_1 - N_3$ into the DOFs expression, yielding
\begin{equation}
\label{cn2}
\begin{aligned}
f_1(N_3)&\sx \triangleq\sx (-1/(4N_1^2) - 2N_1 + 2)N_3^2 + (-15/(8N_1^2) + 3/(2N_1) \\
&- 1/(4N_1^3) - 2N + 1/(4N_1^2N) + 2N_1N)N_3 - 2-2N \\
&- N_1^3/4 + 3N_1 - 7N_1^2/4 + N_1^2N/4 + 2N_1N.
\end{aligned}
\end{equation}

If \(N_1\) is fixed and the total number of physical sensors \(N\) is given,
the unknown parameter in (\ref{cn2}) is only \(N_3\),
and the problem reduces to maximizing a quadratic function \(f_1(N_3)\) with respect to \(N_3\).
Since this is a convex optimization problem, the maximum of \(f_1(N_3)\) occurs where its first derivative equals zero \cite{Krentel1986}.
The first derivative of \(f_1(N_3)\) is given by
\begin{equation}\nonumber
\begin{aligned}
\frac{\partial f_1(N_3)}{\partial N_3}\sx& =\sx 2(-1/(4N_1^2) - 2N_1 + 2)N_3+ (-15/(8N_1^2) \\
&+ 3/(2N_1) \sx-\sx 1/(4N_1^3) - 2N + 1/(4N_1^2N) + 2N_1N).
\end{aligned}
\end{equation}

When $\frac{\partial f_1(N_3)}{\partial N_3}=0$, the value of $N_3$ can be solved as follows
$N_3=\frac{-\frac{1}{4}N_1^3+(\frac{N}{4}-\frac{15}{8})N_1^2+(2N+\frac{3}{2})N_1-2N}{ \frac{1}{2}N_1^2 + 4N_1 - 4}$.

Since $N_3 \in \mathbb{N}_+$, $N_3=\lceil\frac{-\frac{1}{4}N_1^3+(\frac{N}{4}-\frac{15}{8})N_1^2+(2N+\frac{3}{2})N_1-2N}{ \frac{1}{2}N_1^2 + 4N_1 - 4}\rfloor$.
Further, we can get
$N_2=\lceil\frac{-\frac{1}{4}N_1^3+(\frac{N}{4}-\frac{17}{8})N_1^2+(2N+\frac{5}{2})N_1+2N}{ \frac{1}{2}N_1^2 + 4N_1 - 4}\rfloor$.

Next, we discuss the value of \(N_1\). Based on the previous analysis,
the corresponding values of \(N_2\), \(N_3\), \(M_1\) and \(M_2\) for each \(N_1\) can be determined.
Substituting these variables into (\ref{cn2}) yields
\begin{equation}\nonumber
\label{cn3}
\begin{aligned}
&h_1(N_1)\sx \triangleq\sx -\sx N_1^8/(64\varrho^2) \sx+\sx (N/(32\varrho^2)\sx-\sx 19/(64\varrho^2) )N_1^7\sx+\sx (1/(16\varrho)\sx\\
&-\sx 365/(256\varrho^2)\sx-\sx N^2/(64\varrho^2) \sx+\sx 39N/(64\varrho^2))N_1^6\sx+\sx (27/(16\varrho^2)\\
&\sx+\sx 11/(16\varrho) \sx-\sx 5N^2/(16\varrho^2)\sx-\sx N/(8\varrho)\sx+\sx 95N/(32\varrho^2))N_1^5 \\
&\sx+\sx (3597/(256\varrho^2)\sx+\sx 33/(64\varrho)\sx+\sx N^2/(16\varrho) \sx-\sx 99N^2/(64\varrho^2)\\
&\sx-\sx 251N/(64\varrho^2) \sx-\sx 23N/(16\varrho))N_1^4 \sx+\sx ( 9N^2/(4\varrho^2)\sx-\sx1/4 \\
&\sx-\sx 223/(32\varrho) \sx-\sx 657/(32\varrho^2)\sx+\sx 3N^2/(4\varrho)\sx-\sx 491N/(16\varrho^2) \\
&\sx-\sx 19N/(16\varrho))N_1^3 \sx+\sx (273/(64\varrho) \sx+\sx 117/(16\varrho^2)\sx+\sx 399N/(8\varrho^2) \\
&\sx+\sx 11N^2/(16\varrho) \sx+\sx 67N^2/(4\varrho^2)\sx+\sx 123N/(8\varrho)\sx-\sx 3/4 \sx+\sx N/4)N_1^2\\
& \sx+\sx (3/(16\varrho) \sx-\sx 17N^2/(2\varrho)\sx-\sx 30N^2/\varrho^2 \sx-\sx 39N/(2\varrho^2) \sx-\sx 85N/(8\varrho)\\
&\sx+\sx 15/4 \sx+\sx N)N_1 \sx-\sx N/(4\varrho) \sx+\sx 13N^2/\varrho^2 \sx+\sx 13N^2/(2\varrho) \sx-\sx 13N/4 \sx-\sx 11/4,
\end{aligned}
\end{equation}
where $\varrho\sx=\sx 1/2N_1^2 + 4N_1 - 4$. At this stage, our goal is to maximize the value of \(h_1(N_1)\).
However, since \(h_1(N_1)\) is a cubic function of the variable \(N_1\),
finding its maximum is challenging due to the presence of multiple extrema \cite{Krentel1986}.
Therefore, alternative methods are sought to determine the maximum of \(h_1(N_1)\).
Further analysis shows that \(N_1\) ranges from 2 to \(N\).
Therefore, Algorithm 1 is proposed to search for the maximum value of $h_1(N_1)$ with respect to $N_1$,
and the corresponding maximum aperture of FOHA(NA) is obtained.

Based on the above analysis, the number of physical sensors $N_1$, $N_2$ and $N_3$ in subarray 1, 2 and 3,
corresponding to the maximum DOFs of FOHA(NA), are determined.

\begin{table}[h]
\begin{center}
\label{tab1}
\renewcommand{\arraystretch}{0.8}
\begin{tabular}{ l }
\hline
\textbf{Algorithm 1: Search for the optimal array parameters } \\
\ \ \ \ \ \ \ \ \ \ \ \ \ \ \ \ \ \textbf{of FOHA(NA)}  \\
\hline
\bf{Input}: $N$. \\
Initialization: $\text{DOFs}^*=$ 0, $N_1^*=$ 0, $N_2^*=$ 0, $N_3^*=$ 0. \\
\ \ \ \ for $N_1=1$ to $N$.   \\
\ \ \ \ \ \ \ \ $N_2=\lceil\frac{-\frac{1}{4}N_1^3+(\frac{N}{4}-\frac{17}{8})N_1^2+(2N+\frac{5}{2})N_1+2N}{ \frac{1}{2}N_1^2 + 4N_1 - 4}\rfloor$.\\
\ \ \ \ \ \ \ \ $N_3=\lceil\frac{-\frac{1}{4}N_1^3+(\frac{N}{4}-\frac{15}{8})N_1^2+(2N+\frac{3}{2})N_1-2N}{ \frac{1}{2}N_1^2 + 4N_1 - 4}\rfloor$. \\
\ \ \ \ \ \ \ \ $M_1=\lceil\frac{N_1}{2}\rceil$, $M_2=\lfloor\frac{N_1}{2}\rfloor$. \\
\ \ \ \ \ \ \ \ $\delta_1=2M_1M_2+3M_1-2$, $\eta_1=M_1M_2+2M_1-1$.\\
\ \ \ \ \ \ \ \ $\lambda_1=M_1M_2+2M_1-2$, $\lambda_2=M_1M_2+M_1-1$.\\
\ \ \ \ \ \ \ \ $\delta_2\sx=\sx(3N_2\sx+\sx 4)M_1M_2\sx+\sx(6N_2\sx+\sx 4)M_1- 3N_2\sx -\sx 3$.\\
\ \ \ \ \ \ \ \ $\eta_2\sx=\sx(2N_2\sx+\sx 3)M_1M_2\sx+\sx(4N_2\sx+\sx 3)M_1- 2N_2\sx -\sx 2$.\\
\ \ \ \ \ \ \ \ $\lambda_3= \delta_1+\eta_1(N_2-1)$, $\lambda_4= \delta_1+\eta_1(N_2-1)+\lambda_2$.\\
\ \ \ \ \ \ \ \ $E=\delta_2 + \eta_2(N_3 - 1) + \lambda_3$. \\
\ \ \ \ \ \ \ \ $\text{DOFs}=2E+1$. \\
\ \ \ \ \ \ \ \ if \ DOFs $>$ $\text{DOFs}^*$\\
\ \ \ \ \ \ \ \ \ \ \ \ $\text{DOFs}^* =$ DOFs.\\
\ \ \ \ \ \ \ \ \ \ \ \ $N_1^*=N_1$, $N_2^*=N_2$, $N_3^*=N_3$.\\
\ \ \ \ \ \ \ \ end\\
\ \ \ \ end \\
\textbf{Output}: $N_1^*$, $N_2^*$, $N_3^*$, $\text{DOFs}^*$ $\rhd$ optimal array parameters.\\
\hline
\end{tabular}
\end{center}
\end{table}
\end{proof}

\subsection{Generator With Concatenated Nested Array}
The FOHA(CNA) is systematically designed by appropriately placing the sensors of three subarrays, as illustrated in Fig. \ref{fig:stru},
where subarray 1 is a concatenated nested array (CNA) \cite{Robin2017} with $N_1$ sensors,
and subarrays 2 and 3 are SLAs with large inter-sensor spacings.
\subsubsection{Structure of the Proposed FOHA(CNA)}
\begin{definition}
\label{def:cna}
The FOHA(CNA) consists of three subarrays with $N \sx=\sx N_1 + N_2 + N_3$ sensors,
where \(N_i\) is the number of sensors in the subarray $i$ \((i=1,2,3)\).
The sensor positions, denoted by \(\mathbb{A}_1\), \(\mathbb{A}_2\) and \(\mathbb{A}_3\), are given as follows
\begin{equation}\nonumber
\label{wh5}
\begin{aligned}
&\mathbb{P}=\mathbb{A}_1\cup\mathbb{A}_2\cup\mathbb{A}_3,\\
&\mathbb{A}_1\sx=\sx \{ 0:M_1\sx-\sx 1 \} \cup\{ M_1:M_1\sx+\sx 1:M_1\sx+\sx(M_1\sx+\sx 1)(M_2\sx-\sx 1) \}\\
&\ \ \ \ \ \ \cup \{M_1\sx+\sx(M_1\sx+\sx 1)(M_2\sx-\sx 1)\sx+\sx 1:2M_1\sx+\sx (M_1\sx+\sx 1)(M_2\sx-\sx 1) \},\\
&\mathbb{A}_2\sx=\sx \{\delta_1: \eta_1: \delta_1+\eta_1(N_2-1)\},\\
&\mathbb{A}_3\sx=\sx \{\delta_2: \eta_2: \delta_2+\eta_2(N_3-1)\},\\
&\delta_1\sx=\sx 6M_1+3(M_1+1)(M_2-1)+1,\\
&\eta_1\sx=\sx 4M_1+2(M_1+1)(M_2-1)+1,\\
&\delta_2\sx=\sx(12N_2+8)M_1+(6N_2+4)(M_1+1)(M_2-1)\sx+\sx 3N_2\sx +\sx 1,\\
&\eta_2\sx=\sx(8N_2+6)M_1+(4N_2+3)(M_1+1)(M_2-1)\sx +\sx 2N_2\sx +\sx 1.
\end{aligned}
\end{equation}
where $M_1$ and $M_2$ represent the number of physical sensors in subarray 1 and 2 of $\mathbb{A}_1$.
The structure of FOHA(CNA) is illustrated in Fig. \ref{fig:stru} (b).
\end{definition}

\subsubsection{Consecutive Segment of the FOHA(CNA)}
\begin{theorem}
FOHA(CNA) designed based on the FODCA is hole-free.
\end{theorem}
\begin{proof}
According to Lemma \ref{lea:dof}, FOHA(CNA) designed based on FODCA is hole-free provided that conditions (I), (II) and (III) are satisfied.
Regarding condition (I), we have
\begin{equation}\nonumber
\begin{aligned}
\lambda_2 = 2M_1+(M_1+1)(M_2-1)= \max(\mathbb{A}_1),
\end{aligned}
\end{equation}
which satisfies the requirement.

Similarly, for condition (II),
\begin{equation}\nonumber
\begin{aligned}
&\lambda_1\sx=\sx 4M_1\sx+\sx 2(M_1\sx+\sx 1)(M_2\sx-\sx 1) \sx>\sx \max(\mathbb{A}_1)\sx=\sx 2M_1\sx+\sx (M_1\sx+\sx 1)(M_2\sx-\sx 1),
\end{aligned}
\end{equation}
ensuring that condition (II) is met.

For condition (III),
\begin{equation}\nonumber
\begin{aligned}
&\lambda_2=2M_1+(M_1+1)(M_2-1) = \frac{\lambda_1}{2},
\end{aligned}
\end{equation}
which also holds.

Therefore, by Lemma \ref{lea:dof}, FOHA(CNA) designed based on FODCA is hole-free when conditions (I), (II) and (III) are satisfied.
\end{proof}

\subsubsection{The Maximum DOFs of FOHA(CNA) With the Given Number of Physical Sensors}
~\par
The DOFs of FOHA(CNA) can be further increased by optimizing the distribution of physical sensors among subarrays 1, 2 and 3
for a given number of physical sensors.
\begin{theorem}
To maximize the DOFs of FOHA (CNA) with a fixed total number of sensors,
the number of sensors for each subarray is set as
\begin{equation}\nonumber
\begin{cases}
N_1=Caculated\ by\ Algorithm\ 2,\\
N_2=\lceil\frac{-\frac{1}{2}N_1^3+(\frac{N}{2}-\frac{23}{8})N_1^2+(3N+\frac{9}{4})N_1-\frac{3}{2}N +\frac{1}{8}}{ N_1^2 +6N_1 - 3}\rfloor,\\
N_3=\lceil\frac{-\frac{1}{2}N_1^3+(\frac{N}{2}-\frac{25}{8})N_1^2+(3N+\frac{3}{4})N_1-\frac{3}{2}N -\frac{1}{8}}{ N_1^2 +6N_1 - 3}\rfloor.
\end{cases}
\end{equation}

For subarray 1, to maximize the DOFs, the parameters \(M_1\) and \(M_2\) are chosen as
\begin{equation}
M_1=\lceil\frac{N_1-1}{4}\rfloor,\ M_2=\lceil\frac{N_1+1}{2}\rfloor.
\end{equation}
\end{theorem}
\begin{proof}
%
Similar to problem (P1) and (P2) for FOHA(NA),
substituting \(N_2 = N - N_1 - N_3\) into the DOFs expression yields
\begin{equation}
\label{fo8}
\begin{aligned}
f_2(N_3)&\sx \triangleq\sx ( 3/2\sx-\sx 3N_1 \sx-\sx 1/(2N_1^2))N_3^2 \sx+\sx (3N_1N \sx+\sx 1/(2N_1^2N) \sx-\sx 1/8 \\
&\sx+\sx 3/(4N_1) \sx-\sx 25/(8N_1^2) \sx-\sx 3/(2N) \sx-\sx 1/(2N_1^3))N_3 \sx-\sx 3N/2 \sx\\
&-\sx 7/4 \sx+\sx 3N_1 \sx+\sx N_1^2N/2 \sx-\sx 11N_1^2/4 \sx+\sx 3N_1N \sx-\sx N_1^3/2.
\end{aligned}
\end{equation}

Similar to the function \(f_1(N_3)\) in (\ref{cn2}), the maximum of the quadratic function \(f_2(N_3)\) with respect to \(N_3\) can be obtained by solving
\(\frac{\partial f_2(N_3)}{\partial N_3} = 0\), which yields
$N_3 = \frac{-\frac{1}{2}N_1^3 + (\frac{N}{2} - \frac{25}{8})N_1^2 + (3N + \frac{3}{4})N_1 - \frac{3}{2}N - \frac{1}{8}}{N_1^2 + 6N_1 - 3}$.
Since \(N_3 \in \mathbb{N}_+\), we take the nearest integer
$N_3 = \lceil \frac{-\frac{1}{2}N_1^3 + (\frac{N}{2} - \frac{25}{8})N_1^2 + (3N + \frac{3}{4})N_1 - \frac{3}{2}N - \frac{1}{8}}{N_1^2 + 6N_1 - 3} \rfloor$.
Accordingly, \(N_2\) can be expressed as
$N_2 = \lceil \frac{-\frac{1}{2}N_1^3 + (\frac{N}{2} - \frac{23}{8})N_1^2 + (3N + \frac{9}{4})N_1 - \frac{3}{2}N + \frac{1}{8}}{N_1^2 + 6N_1 - 3} \rfloor$.

Next, we consider the value of \(N_1\).
Based on the above analysis, the corresponding values of \(N_2\), \(N_3\), \(M_1\) and \(M_2\) can be determined for \(N_1\).
Substituting these into (\ref{fo8}) yields
\begin{equation}\nonumber
\label{fo9}
\begin{aligned}
&h_2(N_1)\sx \triangleq\sx (- 1/(4\varsigma) \sx-\sx N_3/(4\varsigma))N_1^5 \sx+\sx (N/(4\varsigma) \sx-\sx 49/(16\varsigma)\\
&\sx-\sx 49N_3/(16\varsigma)\sx+\sx NN_3/(4\varsigma))N_1^4\sx+\sx (3N/\varsigma\sx-\sx N_3/2\sx-\sx 33/(4\varsigma)\\
&\sx-\sx 33N_3/(4\varsigma) \sx+\sx 3NN_3/\varsigma \sx-\sx 1/2)N_1^3 \sx+\sx (55/(8\varsigma)\sx-\sx 21N_3/8\sx+\sx N/2\\
&\sx-\sx 11/4 \sx+\sx 55N_3/(8\varsigma)\sx+\sx 15N/(2\varsigma)\sx+\sx NN_3/2 \sx+\sx 15NN_3/(2\varsigma))N_1^2 \\
&\sx+\sx (15N_3/4\sx-\sx 9NN_3/\varsigma \sx-\sx 3/(2\varsigma) \sx+\sx 3 \sx+\sx 3N \sx-\sx 3N_3/(2\varsigma)\sx-\sx 9N/\varsigma \\
&\sx+\sx 3NN_3)N_1 \sx-\sx 3N/2 \sx-\sx 3NN_3/2 \sx-\sx 7/4 \sx+\sx 3/(16\varsigma)\sx +\sx 3N_3/(16\varsigma) \\
&\sx+\sx 9N/(4\varsigma) \sx+\sx 9NN_3/(4\varsigma) \sx-\sx 13N_3/8,
\end{aligned}
\end{equation}
where $\varsigma=-N_1^2 - 6N_1 + 3$.
At this stage, the objective is to maximize the function \(h_2(N_1)\).
Similar to \(h_1(N_1)\), Algorithm~2 is developed to search for the maximum value of $h_2(N_1)$ with respect to $N_1$ in $[2, N]$,
and the corresponding maximum aperture of FOHA(CNA) is obtained.

Based on the above analysis, the number of physical sensors $N_1$, $N_2$ and $N_3$ in subarray 1, 2 and 3,
corresponding to the maximum DOFs of FOHA(CNA), are determined.
\end{proof}

\begin{table}
\begin{center}
\label{tab1}
\renewcommand{\arraystretch}{0.8}
\begin{tabular}{ l }
\hline
\textbf{Algorithm 2: Search for the optimal array parameters } \\
\ \ \ \ \ \ \ \ \ \ \ \ \ \ \ \ \ \textbf{of FOHA(CNA)}  \\
\hline
\bf{Input}: $N$. \\
Initialization: $\text{DOFs}^*=$ 0, $N_1^*=$ 0, $N_2^*=$ 0, $N_3^*=$ 0. \\
\ \ \ \ for $N_1=1$ to $N$.   \\
\ \ \ \ \ \ \ \ $N_2=\lceil\frac{-\frac{1}{2}N_1^3+(\frac{N}{2}-\frac{23}{8})N_1^2+(3N+\frac{9}{4})N_1-\frac{3}{2}N +\frac{1}{8}}{ N_1^2 +6N_1 - 3}\rfloor$.\\
\ \ \ \ \ \ \ \ $N_3=\lceil\frac{-\frac{1}{2}N_1^3+(\frac{N}{2}-\frac{25}{8})N_1^2+(3N+\frac{3}{4})N_1-\frac{3}{2}N -\frac{1}{8}}{ N_1^2 +6N_1 - 3}\rfloor$. \\
\ \ \ \ \ \ \ \ $M_1=\lceil\frac{N_1-1}{4}\rfloor$, $M_2=\lceil\frac{N_1+1}{2}\rfloor$. \\
\ \ \ \ \ \ \ \ $\delta_1=6M_1+3(M_1+1)(M_2-1)+1$.\\
\ \ \ \ \ \ \ \ $\lambda_1\sx=\sx 4M_1\sx+\sx 2(M_1\sx+\sx 1)(M_2\sx-\sx 1)$.\\
\ \ \ \ \ \ \ \ $\lambda_2\sx=\sx 2M_1\sx+\sx (M_1\sx+\sx 1)(M_2\sx-\sx 1)$.\\
\ \ \ \ \ \ \ \ $\eta_1=4M_1+2(M_1+1)(M_2-1)+1$.\\
\ \ \ \ \ \ \ \ $\delta_2\sx=\sx(12N_2\sx+\sx 8)M_1\sx+\sx(6N_2\sx+\sx 4)(M_1\sx +\sx 1)(M_2\sx -\sx 1)\sx+\sx 3N_2\sx +\sx 1$.\\
\ \ \ \ \ \ \ \ $\eta_2\sx=\sx(8N_2\sx +\sx 6)M_1\sx+\sx (4N_2\sx+\sx 3)(M_1\sx+\sx 1)(M_2\sx-\sx 1)\sx +\sx 2N_2\sx +\sx 1$.\\
\ \ \ \ \ \ \ \ $\lambda_3= \delta_1+\eta_1(N_2-1)$, $\lambda_4= \delta_1+\eta_1(N_2-1)+\lambda_2$.\\
\ \ \ \ \ \ \ \ $E=\delta_2 + \eta_2(N_3 - 1) + \lambda_3$. \\
\ \ \ \ \ \ \ \ $\text{DOFs}=2E+1$. \\
\ \ \ \ \ \ \ \ if \ DOFs $>$ $\text{DOFs}^*$\\
\ \ \ \ \ \ \ \ \ \ \ \ $\text{DOFs}^* =$ DOFs.\\
\ \ \ \ \ \ \ \ \ \ \ \ $N_1^*=N_1$, $N_2^*=N_2$, $N_3^*=N_3$.\\
\ \ \ \ \ \ \ \ end\\
\ \ \ \ end \\
\textbf{Output}: $N_1^*$, $N_2^*$, $N_3^*$, $\text{DOFs}^*$ $\rhd$ optimal array parameters.\\
\hline
\end{tabular}
\end{center}
\end{table}

\subsection{Examples of FOHA}
The array structures of FOHA(NA) and FOHA(CNA) designed using the proposed scheme with \(N = 9\) physical sensors are illustrated in Fig.~\ref{fig:eg}.
The corresponding DOFs are 397 and 413, respectively,
which are higher than the DOFs of 307 for FO-Fractal(NA) in \cite{Yang2023} and 375 for SD-FOSA(CNA-NA) in \cite{ChenH2025}.
\begin{figure*}
  \centering
  \subfigure[]{
    \includegraphics[scale=0.435]{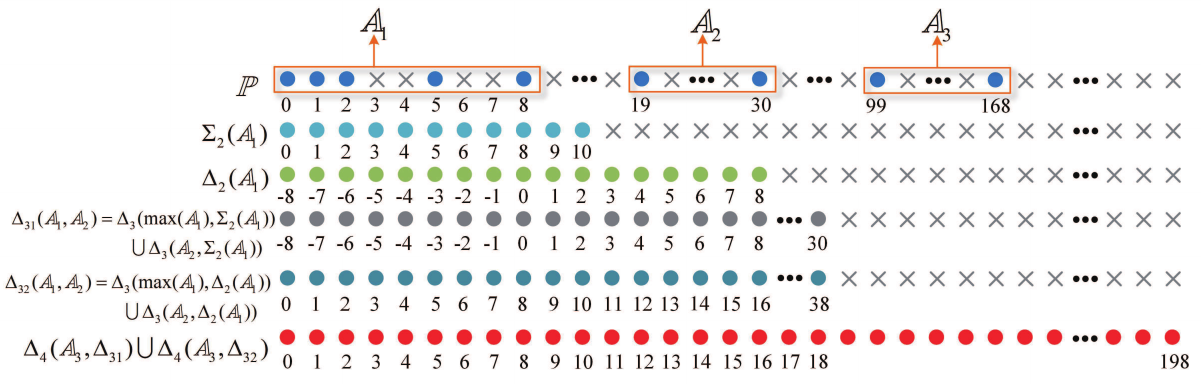}}
  \hspace{-0.1in} 
  \subfigure[]{
    \includegraphics[scale=0.435]{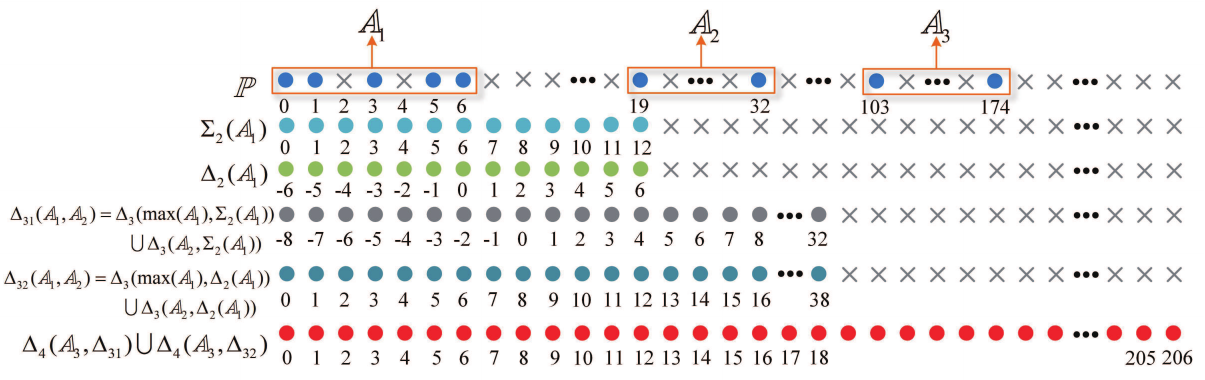}}
 \caption{\label{1} Structure of FOHA with different generators when $N=9$. (a) FOHA(NA). (b) FOHA(CNA). }
\label{fig:eg}
\end{figure*}

\section{Necessary and Sufficient Conditions for Signal Reconstruction}
In array signal processing, the problem of signal reconstruction can be fundamentally reformulated as the estimation of the DOA.
Specifically, for the received signal $\mathbf{x}(t)$ with an $N$ sensors array denoted as (\ref{w2}),
and $\mathbf{s}(t) \in \mathbb{C}^D$ is the source signals vector from $D$ narrowband far-field sources,
$\boldsymbol{\theta}_{[1:D]} = [\theta_1, \dots, \theta_D]^T $ denoting the unknown DOAs,
and $\mathbf{A}(\boldsymbol{\theta}) = [\mathbf{a}(\theta_1), \dots, \mathbf{a}(\theta_D)] \in \mathbb{C}^{N \times D} $ is the array manifold matrix,
composed of steering vectors determined by the array geometry and $\boldsymbol{\theta}_{[1:D]}$.
Once $\boldsymbol{\theta}_{[1:D]}$ are accurately estimated
and the array manifold matrix $\mathbf{A}(\boldsymbol{\theta})$ satisfies the rank condition,
the source signals vector $\mathbf{s}(t)$ can be reconstructed from the received signal $\mathbf{x}(t)$.
At this point, whether the signal can be reconstructed is equivalent to whether there is a one-to-one correspondence between $\mathbf{x}(t) $ and $\boldsymbol{\theta}_{[1:D]}$.
This requires several conditions are satisfied as follows \cite{Krim1996}, \cite{Tuncer2009}, \cite{ThompsonAR2017}.

Firstly, the signal characteristics must be taken into account.
The signal should satisfy the narrowband assumption, i.e., $B \ll f_c $,
where $B$ is the signal bandwidth and $f_c$ is the carrier frequency.
Secondly, the relationship between the array dimension and the number of source signals is crucial.
The number of sensors $N$ should be greater than or equal to the number of source signals $D$, i.e., $N \geq D$.
Finally, the array manifold matrix $\mathbf{A}$ should have full column rank, i.e., $\text{rank}(\mathbf{A}) = D $ \cite{Horn2012}.
Under above conditions, the linear system (\ref{w2}) ensures a unique solution, establishing the one-to-one correspondence between $\mathbf{x}(t)$  and $\mathbf{\theta}_{[1:D]}$.
Finally, the array geometry plays a vital role.
When the inter-spacing $d$ exceeds half the wavelength ($\lambda/2$) in any array,
the problem of ambiguous angles may arise,
causing the array steering matrix $\mathbf{A}$ to lose full rank.
This leads to a situation where the one-to-one correspondence between $\mathbf{x}(t)$ and $\boldsymbol{\theta}_{[1:D]}$ is no longer satisfied,
making it impossible to uniquely reconstruct the source signals.
Therefore, the geometry and the inter-sensor spacing $d$ should satisfy a certain threshold to prevent the occurrence of grid flap problems,
ensuring that DOA estimation remains unambiguous.

To sum up, there exists necessary and sufficient conditions for signal reconstruction as follows.
\begin{remark}
The function \(\mathrm{LCM}(\beta_1,\beta_2)\) for calculating the least common multiple of two numbers \(\beta_1\) and \(\beta_2\) is proposed in Algorithm 3. Similarly, the function \(\mathrm{LCM\mbox{-}multi}(\text{numbers})\) for the least common multiple of multiple numbers, and the function \(\mathrm{LCM\mbox{-}seq}(\beta_1,\beta_2,N)\) for that of an \(N\)-term arithmetic sequence with first term \(\beta_1\) and common difference \(\beta_2\), are also proposed in Algorithm 3.
\end{remark}

\begin{theorem}
For any array $\mathbb{P}=\{ p_1, p_2, \dots, p_N \}$, the sufficient and necessary condition for signal reconstruction is
\begin{equation}
\label{st14}
\operatorname{LCM-multi}\left( \frac{\lambda}{p_{l_1}}, \frac{\lambda}{p_{l_2}}, \dots, \frac{\lambda}{p_{l_N}} \right) \geq 2,
\end{equation}
that is,
$\min_{c_1, c_2, \dots, c_N \in \mathbb{Z}^+} \left\{ \frac{c_1 \lambda}{p_{l_1}} = \frac{c_2 \lambda}{p_{l_2}} = \dots = \frac{c_N \lambda}{p_{l_N}} \right\} \geq 2$.
\end{theorem}

\begin{proof}
\textbf{Sufficiency:}
When (\ref{st14}) is satisfied,
it ensures that the array can uniquely resolve all directions within the angular field of view, without aliasing.
Therefore, for any $\theta_i \in [-\frac{\pi}{2}, \frac{\pi}{2}]$, the mapping from $\theta_i$ to the steering vector $\mathbf{a}(\theta_i)$ is injective.

Hence, no two distinct $\theta_i \neq \theta_i'$ can yield the same steering vector, i.e., $\mathbf{a}(\theta_i) \neq \mathbf{a}(\theta_i')$,
and the array manifold $\mathbf{A}$ is uniquely determined by $\mathbf{\theta}_{[1:D]}$.
Therefore, there is a a unique correspondence between the received signal vector $\mathbf{x}(t)$ and $\mathbf{\theta}_{[1:D]}$
when a full-rank signal matrix $\mathbf{S}$ is given,
ensuring the source signals $\mathbf{s}(t)$ can be reconstructed from the received signal $\mathbf{x}(t)$.

\textbf{Necessity:}
When the signal is uniquely reconstructable, i.e., there exists a bijective mapping between $\mathbf{x}(t)$ and $\mathbf{\theta}_{[1:D]}$,
which implies that the steering vector satisfies
\[
\mathbf{a}(\theta_i) \neq \mathbf{a}(\theta_i'), \quad \forall\ \theta_i \neq \theta_i' \in \left[-\frac{\pi}{2}, \frac{\pi}{2}\right].
\]

By using the method of contradiction, assuming
\begin{equation}
\label{st13}
\mathcal{F} = \operatorname{LCM-multi}\left( \frac{\lambda}{p_{l_1}}, \dots, \frac{\lambda}{p_{l_N}} \right) < 2.
\end{equation}

Then, spatial aliasing arising from periodic spatial under-sampling occurs.
There may exist distinct $\theta_i \neq \theta_i'$ such that for all $n = 1,\dots,N$ and some integers $k_{l_n} \in \mathbb{Z}$ \cite{ThompsonAR2017},
\[
2\pi \frac{p_{l_n}}{\lambda} \sin(\theta_i) = 2\pi \left( \frac{p_{l_n}}{\lambda} \sin(\theta_i') + k_{l_n} \right),
\]
or equivalently $\sin(\theta_i) \equiv \sin(\theta_i') \mod \left( \frac{\lambda}{p_{l_n}} \right)$.

This indicates that two distinct angles $\theta_i \ne \theta_i'$ yield identical steering vectors, i.e., $\mathbf{a}(\theta_i) = \mathbf{a}(\theta_i')$,
thereby violating the injective mapping from angle $\theta$ to its steering vector $\mathbf{a}(\theta)$.
To eliminate angular ambiguity, the total unambiguous angular range,
governed by the least common multiple of the spatial sampling intervals,
must span at least one full period, i.e.,
$\mathcal{F} \geq 2$, which contradicts the constraint in (\ref{st13}).
Therefore, the least common multiple condition is both necessary and sufficient for guaranteeing a unique mapping between the $\mathbf{x}(t)$  and $\mathbf{\theta}_{[1:D]}$
to reconstruct the source signals $\mathbf{s}(t)$.
\end{proof}

For FOHA with $N$ sensors located at the set $\mathbb{P}=\mathbb{A}_1\cup \mathbb{A}_2 \cup \mathbb{A}_3$,
the necessary and sufficient conditions for signal reconstruction are given by the following theorem.

\begin{theorem}
When the physical sensor position sets \(\mathbb{P} \) of FOHA(NA) and FOHA(CNA)
are defined as in Definition~\ref{def:na} and \ref{def:cna},
the source signals \(\mathbf{s}(t)\) can be uniquely reconstructed from \(\mathbf{x}(t)\) and \(\boldsymbol{\theta}_{[1:D]}\)
if and only if there exist coefficients \(\{c_n\}_{n=1}^N\) and a positive integer \(k\) such that
\begin{equation}\label{eq:coef_condition_new}
\begin{cases}
c_{n_1} = \dfrac{k \cdot p_{l_{n_1}}}{\epsilon_1}, \quad n_1 = 1,2,...,N_1,\\
c_{n_2} = \dfrac{k \cdot p_{l_{n_2}}}{\epsilon_2}, \quad n_2 = N_1\sx+\sx 1, ..., N_1\sx+\sx N_2,\\
c_{n_3} = \dfrac{k \cdot p_{l_{n_3}}}{\epsilon_3}, \quad n_3 = N_1\sx+\sx N_2\sx+\sx 1, ..., N_1\sx+\sx N_2\sx+\sx N_3,\\
k \geq \dfrac{2\cdot \operatorname{LCM-multi}(\epsilon_1,\epsilon_2,\epsilon_3)}{\lambda},
\end{cases}
\end{equation}
where $\zeta_1=\operatorname{LCM-seq}(1, 1, M_1- 1) $, $\zeta_2= \operatorname{LCM-seq}(2M_1- 1, M_1, M_2M_1+ M_1- 1)$,
$\zeta_3=\operatorname{LCM-seq}(M_1, M_1+1, M_2) $, $\zeta_4= \operatorname{LCM-seq}(M_1+(M_1+1)(M_2-1)+1, 1, M_1)$,
\begin{align*}
\epsilon_1 &\sx=
\begin{cases}
\sx \frac{\zeta_1 \cdot \zeta_2}{\text{gcd}(\zeta_1,\zeta_2)},\ \text{for FOHA(NA)},\\
\sx \frac{\frac{\zeta_1 \cdot \zeta_3}{\text{gcd}(\zeta_1,\zeta_3)}\cdot \zeta_4}{\text{gcd}(\frac{\zeta_1 \cdot \zeta_3}{\text{gcd}(\zeta_1,\zeta_3)}, \zeta_4)},\ \text{for FOHA(CNA)},
\end{cases}\\
\epsilon_2 &= \operatorname{LCM-seq}(\delta_1, \eta_1, N_2),\
\epsilon_3 = \operatorname{LCM-seq}(\delta_2, \eta_2, N_3).
\end{align*}
\end{theorem}

\begin{proof}
\textbf{Sufficiency:}
When there exist coefficients $\{c_n\}_{n=1}^N$ and an integer $k$ satisfying the conditions in (\ref{eq:coef_condition_new}).

We analyze the three subsets separately:

\emph{1) Analysis of $\mathbb{A}_1$:}
\textbf{Case 1 (FOHA(NA)):}
The elements \(\{p_{l_1}, \ldots, p_{l_{N_1}}\}\) in \(\mathbb{A}_1\) form two arithmetic sub-sequences.
For the first sub-sequence, the initial term is 0 and must be excluded.
Therefore, the effective sequence starts from 1 with a common difference of 1.
According to number theory results~\cite{Xiao2023}, the corresponding least common multiple is given by
$\zeta_1 = \operatorname{LCM-seq}(1, 1, M_1 - 1)$,
which can be computed using Algorithm~3.

Similarly, for the second sub-sequence, the least common multiple is expressed as
$\zeta_2 = \operatorname{LCM-seq}(2M_1 - 1, M_1, M_2M_1 + M_1 - 1)$.
Then, the overall least common multiple for \(\mathbb{A}_1\) is calculated as
$\epsilon_1 = \frac{\zeta_1 \cdot \zeta_2}{\gcd(\zeta_1, \zeta_2)}$.
The corresponding coefficients are assigned as
$c_{n_1} = \frac{k \cdot p_{l_{n_1}}}{\epsilon_1}$.

\textbf{Case 2 (FOHA(CNA)):}
Following a similar procedure, the least common multiple of \(\mathbb{A}_1\) for FOHA(CNA) can be computed using its specific structure,
resulting in the corresponding \(\epsilon_1\).

\emph{2) Analysis of $\mathbb{A}_2$:}
For the second subarray \(\mathbb{A}_2 \sx=\sx \{p_{l_{N_1+1}}, \ldots, p_{l_{N_1+N_2}}\}\),
as defined in Definition~\ref{def:na}, the least common multiple is
$\epsilon_2 \sx=\sx \operatorname{LCM-seq}(\delta_1, \eta_1, N_2)$.
The corresponding coefficients are
$c_{n_2} \sx=\sx \frac{k \cdot p_{l_{n_2}}}{\epsilon_2}$.

\emph{3) Analysis of $\mathbb{A}_3$:}
Similarly, for the third subarray \(\mathbb{A}_3 = \{p_{l_{N_1+N_2+1}}, \ldots, p_{l_{N_1+N_2+N_3}}\}\),
the least common multiple is
$\epsilon_3 = \operatorname{LCM-seq}(\delta_2, \eta_2, N_3)$,
and the corresponding coefficients are
$c_{n_3} = \frac{k \cdot p_{l_{n_3}}}{\epsilon_3}$.

\emph{4) Overall least common multiple Calculation:}
The total least common multiple across all three subarrays is given by
\begin{align*}
\operatorname{LCM-multi}(\epsilon_1,\epsilon_2,\epsilon_3)\sx = \sx \dfrac{\dfrac{\epsilon_1\sx \cdot\sx \epsilon_2}{\gcd(\epsilon_1,\epsilon_2)} \sx\cdot\sx \epsilon_3}{\gcd( \dfrac{\epsilon_1\sx\cdot\sx \epsilon_2}{\gcd(\epsilon_1,\sx \epsilon_2)}, \epsilon_3 )}.
\end{align*}

\emph{5) Condition for Source Signal Reconstruction:}
To guarantee signal reconstructability, the coefficient \(k\) must satisfy the following lower bound
$k \sx \geq\sx \dfrac{2 \cdot \operatorname{LCM-multi}(\epsilon_1,\epsilon_2,\epsilon_3)}{\lambda}$.

Consequently, the following inequality holds
$\operatorname{LCM-multi}(\dfrac{\lambda}{p_{l_1}}, \dfrac{\lambda}{p_{l_2}}, ..., \dfrac{\lambda}{p_{l_N}}) \sx \geq\sx 2$,
which satisfies the condition in Theorem~3 and ensures successful reconstruction of the source signal \(\mathbf{s}(t)\).

\begin{table}[h]
\begin{center}
\label{tab1}
\renewcommand{\arraystretch}{0.9}
\begin{tabular}{ l }
\hline
\textbf{Algorithm 3: Calculate the least common multiple } \\
\textbf{ \ \ \ \ \ \ \ \ \ \ \ \ \ \ \ \ of sequence} \\
\hline
\bf{Input}: Set $\mathbb{P}$.\\
$\beta_1=\mathbb{P}[1]$, $\beta_2=\mathbb{P}[2]-\mathbb{P}[1]$, $N=length(\mathbb{P})$. \\
\textbf{Function LCM($\boldsymbol{\beta_1}$,$\boldsymbol{\beta_2}$):}\\
\ \ \ \ return ($|\beta_1 \cdot (\beta_1+\beta_2) /gcd(\beta_1,\beta_1+\beta_2)|$)\\
end\\
\textbf{Function LCM-multi(numbers):}\\
\ \ \ \ Initialization: set result = numbers[0]\\
\ \ \ \ for each number in numbers[1:]:\\
\ \ \ \ \ \ \ \ result = LCM(result,number)\\
\ \ \ \ end\\
\ \ \ \ return result\\
end\\
\textbf{Function LCM-seq($\boldsymbol{\beta_1}$,$\boldsymbol{\beta_2}$,$\boldsymbol{N}$):}\\
\ \ \ \ Define sequence as empty list\\
\ \ \ \ for i from 0 to $N-1$:\\
\ \ \ \ \ \ \ \ add ($\beta_1\sx+\sx i\times \beta_2$) to sequence\\
\ \ \ \ end\\
\ \ \ \ return LCM-multi(sequence) $\triangleq$ LCM-seq($\beta_1,\beta_2,N$)\\
end\\
\textbf{Output}: LCM-seq($\beta_1,\beta_2,N$) Least common multiple of $\mathbb{P}$.\\
\hline
\end{tabular}
\end{center}
\end{table}

\textbf{Necessity:}
Conversely, if the source signals $\mathbf{s}(t)$ can be reconstructed according to Theorem~5,
the least common multiple satisfies
$\operatorname{LCM-multi}(\dfrac{\lambda}{p_{l_1}}, \dfrac{\lambda}{p_{l_2}}, ..., \dfrac{\lambda}{p_{l_N}}) \geq 2$.

Thus, it is necessary that $k$ be chosen sufficiently large according to
$k \geq \dfrac{2\cdot \operatorname{LCM-multi}(\epsilon_1,\epsilon_2,\epsilon_3)}{\lambda}$,
with coefficients $\{c_n\}_{n=1}^N$ assigned accordingly as in (\ref{eq:coef_condition_new}).

Hence, the conditions in (\ref{eq:coef_condition_new}) are also necessary.
This completes the proof.
\end{proof}

\section{The Fourth-Order Redundancy}
To facilitate a more comprehensive comparison of the DOFs performance across different FODCA structures,
the concept of \textit{fourth-order redundancy} is introduced as follows.
\begin{definition}
\label{def:redundancy}
According to \cite{ChenYP2021}, the fourth-order redundancy is defined as
\[
R_4:=\frac{\tilde{k}_4(N)}{U_4},
\]
where $\tilde{k}_4(N)=(k_4(N)-1)/2$ denotes the one-sided maximal size of FODCA and $U_4$ is defined in (2).
\end{definition}

The \textit{fourth-order redundancy} quantifies the discrepancy between the actual FODCA \(\mathbb{U}_4\) and the ideal maximal FODCA.
Specifically, when \( R_4 = 1 \), all fourth-order difference terms are distinct and form a consecutive ULA segment.
In contrast, if \( R_4 > 1 \), it indicates either a normalized size of \(\Delta_4\) less than 1, or the presence of holes within \(\Delta_4\).

According to the theorem 3 in reference \cite{ChenYP2021},
we can get the lower bound of $R_4$ is given by a function of $N$.
\begin{theorem}
The fourth-order redundancy $R_4$ of $\mathbb{P}$ satisfies
\begin{equation}
\begin{aligned}
R_4>L_4(N)&:=2 \left( 1+\frac{2}{3\pi} \right)\frac{\tilde{k}_4(N)}{\left( \sx \left(\begin{matrix} \small{N} \\ \small{2} \\ \end{matrix} \right)\sx \right)^2}\\
&=\left( 1+\frac{2}{3\pi} \right)\frac{(N-1)(N^2-N+6)}{N(N+1)^2},
\end{aligned}
\end{equation}
where $\left( \sx \left(\begin{matrix} \small{n} \\ \small{k} \\ \end{matrix} \right) \sx \right)\sx :=\sx \left( \begin{matrix} \small{n+k-1} \\ \small{k} \\ \end{matrix} \right)$
denotes the multiset coefficient.
\end{theorem}

In addition, for the special case of FODCA,
the fourth-order redundancy \( R_4 \) of FOHA(NA) and FOHA(CNA) is analyzed.
The maximum length of the consecutive segment in FODCA is given by
\begin{equation}
\begin{aligned}\nonumber
U_4 = \delta_2 + \eta_2(N_3 - 1) + \lambda_3,
\end{aligned}
\end{equation}
where the values of $\delta_2$, $\eta_2$ and $\lambda_3$ are shown in Definition \ref{def:na} for FOHA(NA) and in Definition \ref{def:cna} for FOHA(CNA).

Therefore, the fourth-order redundancy \( R_4 \) of the FOHA(CNA) can be derived based on Definition~\ref{def:redundancy} as follows
\begin{equation}\nonumber
\begin{aligned}
&R_4=\frac{N^4-2N^3+7N^2-6N}{8U_4}.
\end{aligned}
\end{equation}

\section{Analysis of Mutual Coupling Effect}
To analyze the mutual coupling effect of the FOHA, this section demonstrates that the generated array can inherit
the anti-coupling ability of the generator arrays \cite{Yang2023}.
For FOHA composed of subarrays $\mathbb{A}_1$, $\mathbb{A}_2$ and $\mathbb{A}_3$,
the mutual coupling effect is mainly determined by the dense part of the subarrays $\mathbb{A}_1$
since the gap between the three subarrays is generally
large enough, and the minimum inter-sensor spacings of the subarray $\mathbb{A}_2=\delta_1+\eta_1(N_2-1)$
and $\mathbb{A}_3=\delta_2+\eta_2(N_3-1)$ are $\eta_1$ and $\eta_2$,
which can be almost ignored as a larger
separation leading to less mutual coupling.
In the following, coupling leakage is adopted to quantify the mutual coupling
effect and the relationship between the coupling leakage of
the generated array $\mathbb{P}$ and its generator array $\mathbb{A}_1$ is analyzed.
\begin{proposition}
\label{propo:L}
Consider the subarray $\mathbb{A}_1$ with $N_1$ sensors,
whose coupling leakage is denoted as $L_1$,
while the subarray $\mathbb{A}_2$ and $\mathbb{A}_3$ have $N_2$ and $N_3$ sensors.
When mutual coupling limit $B$ satisfies $B<\max(\mathbb{A}_1)$,
the coupling leakage of FOHA $\mathbb{P}=\mathbb{A}_1\cup\mathbb{A}_2\cup\mathbb{A}_3$ denoted as $L$ satisfies
\begin{equation}
L=\frac{\parallel \mathbf{C}_{\mathbb{A}_1}-diag(\mathbf{C}_{\mathbb{A}_1}) \parallel_F}{\sqrt{\parallel \mathbf{C}_{\mathbb{A}_1}-diag(\mathbf{C}_{\mathbb{A}_1}) \parallel_F^2+L_1^2N_2^2+L_1^2N_3^2}}\times L_1 <L_1,
\end{equation}
where $\mathbf{C}_{\mathbb{A}_1}$ is the mutual coupling matrix of $\mathbb{A}_1$.
\end{proposition}

\begin{proof}
For the subarray $\mathbb{A}_1=\{p_1,p_2,...,p_{N_1}\}$, the $N_1\times N_1$ mutual coupling matrix is
\begin{equation}
\mathbf{C}_{\mathbb{A}_1}=
\begin{pmatrix}
c_0 & c_{|p_2-p_1|} &\cdots & c_{|p_{N_1}-p_1|}\\
c_{|p_1-p_2|} & c_0 & \cdots & c_{|p_{N_1}-p_2|}\\
\vdots & \vdots & \ddots & \vdots\\
c_{|p_1-p_{N_1}|} & c_{|p_2-p_{N_1}|} & \cdots & c_0
\end{pmatrix},
\end{equation}
and the mutual coupling leakage $L_1$ is
\begin{equation}
\begin{aligned}
L_1=\frac{\parallel \mathbf{C}_{\mathbb{A}_1}-diag(\mathbf{C}_{\mathbb{A}_1}) \parallel_F}{\parallel \mathbf{C}_{\mathbb{A}_1}\parallel_F}.
\end{aligned}
\end{equation}

 For the subarray $\mathbb{A}_2=\delta_1+\eta_1(N_2-1)$ and $\mathbb{A}_3=\delta_2+\eta_2(N_2-1)$ whose inter-sensor spacing is at least $\eta_1$,
 since $B<\max(\mathbb{A}_1)$, the mutual coupling matrix can be written as an identity matrix $\mathbf{I}_{N_2}$ and $\mathbf{I}_{N_3}$ since $c_0=1$.

 When $B<\max(\mathbb{A}_1)$, the mutual coupling matrix of FOHA can be calculated as
 \begin{equation}
 \mathbf{C}_{\mathbb{P}}=
 \begin{pmatrix}
\mathbf{C}_{\mathbb{A}_1} & \bf{0} & \bf{0}\\
\bf{0} & \mathbf{C}_{\mathbb{A}_2} & \bf{0}\\
\bf{0} & \bf{0} &\mathbf{C}_{\mathbb{A}_3}
\end{pmatrix}\in \mathbb{C}^{(N_1+N_2+N_3)\times (N_1+N_2+N_3)}.
 \end{equation}

 The corresponding coupling leakage is
 \begin{equation}
 \begin{aligned}
 &L=\frac{\parallel \mathbf{C}_{\mathbb{P}}-\text{diag}(\mathbf{C}_{\mathbb{P}}) \parallel_F}{\parallel \mathbf{C}_{\mathbb{P}}\parallel_F}\\
 &=\frac{\parallel \mathbf{C}_{\mathbb{A}_1}-\text{diag}(\mathbf{C}_{\mathbb{A}_1}) \parallel_F}{\sqrt{ \parallel \mathbf{C}_{\mathbb{A}_1}\parallel_F^2 + \parallel \mathbf{I}_{N_2}\parallel_F^2 + \parallel \mathbf{I}_{N_3}\parallel_F^2}}\\
 &<\frac{\parallel \mathbf{C}_{\mathbb{A}_1}-\text{diag}(\mathbf{C}_{\mathbb{A}_1}) \parallel_F}{\parallel \mathbf{C}_{\mathbb{A}_1}\parallel_F}=L_1.
 \end{aligned}
 \end{equation}

 The relationship between $L$ and $L_1$ is expressed as

\begin{equation}\nonumber
\begin{aligned}
&\frac{1}{L^2}=\frac{\parallel \mathbf{C}_{\mathbb{A}_1}\parallel_F^2 + \parallel \mathbf{I}_{N_2}\parallel_F^2 + \parallel \mathbf{I}_{N_3}\parallel_F^2}{\parallel
\mathbf{C}_{\mathbb{A}_1}-\text{diag}(\mathbf{C}_{\mathbb{A}_1}) \parallel_F^2}\\
&=\frac{1}{L_1^2}+\frac{N_2^2 + N_3^2}{\parallel \mathbf{C}_{\mathbb{A}_1}-\text{diag}(\mathbf{C}_{\mathbb{A}_1}) \parallel_F^2}\\
&=\frac{\parallel \mathbf{C}_{\mathbb{A}_1}-\text{diag}(\mathbf{C}_{\mathbb{A}_1}) \parallel_F^2+L_1^2(N_2^2 + N_3^2)}{L_1^2 \cdot \parallel
\mathbf{C}_{\mathbb{A}_1}-\text{diag}(\mathbf{C}_{\mathbb{A}_1}) \parallel_F^2}\\
\end{aligned}
\end{equation}
\begin{equation}
\begin{aligned}
&\Rightarrow L\sx=\sx \frac{\parallel \mathbf{C}_{\mathbb{A}_1}\sx-\sx \text{diag}(\mathbf{C}_{\mathbb{A}_1}) \parallel_F}{\sqrt{\parallel \mathbf{C}_{\mathbb{A}_1}\sx-\sx \text{diag}(\mathbf{C}_{\mathbb{A}_1}) \parallel_F^2\sx+ L_1^2N_2^2\sx+\sx L_1^2N_3^2}}\sx \times \sx L_1 \sx <\sx L_1.
\end{aligned}
\end{equation}
\end{proof}

\section{PERFORMANCE COMPARISON}
In this section, numerical simulations are conducted to evaluate and compare the performance of FL-NA \cite{Piya2012},
SE-FL-NA \cite{Shen2019}, FO-Fractal(NA) \cite{Yang2023}, FO-SDE(NA) \cite{YangZ2022},
SD-FOSA(NA) \cite{ChenH2025} and SD-FOSA(CNA-NA) \cite{ChenH2025} with the proposed FOHA configurations employing different generators.
The comparison is performed in terms of DOFs, mutual coupling, resolution and root-mean-square error (RMSE),
under varying input signal-to-noise ratio (SNR), number of snapshots and number of sources.
Note that in DOA estimation, the spatial smoothing MUSIC algorithm \cite{Pal22011}, \cite{Liu2015}, \cite{Piya2012}, \cite{You2021}
is widely adopted for subspace-based estimation. In the simulations, all incident sources are assumed to have equal power,
and the number of sources is known a priori.
To quantitatively evaluate the estimation performance, the RMSE is calculated over 1000 independent Monte Carlo trials as follows:
\begin{equation}
\mathrm{RMSE} = \sqrt{\frac{1}{1000D} \sum_{j=1}^{1000} \sum_{i=1}^{D}( \hat{\theta}_i^{j} - \theta_i )^2 },
\end{equation}
where \(\hat{\theta}_i^{j}\) denotes the estimated angle of the \(i^{\text{th}}\) source in the \(j^{\text{th}}\) trial.
Following the methodology in \cite{LiuJ2017}, we focus on the DOFs achieved by each array structure,
rather than the physical aperture, to assess the overall DOA estimation performance.

\subsection{Comparison of the DOFs for Different Arrays}
We compare the DOFs of eight array configurations under a fixed number of physical sensors.
As shown in Fig.~\ref{fig:dof}, it is observed that when \( N \geq 10 \), the proposed FOHA(CNA) achieves higher DOFs than all other seven array structures.
Furthermore, when \( N \geq 12 \), the DOFs of FOHA(NA) also surpass those of the remaining six arrays, with the exception of FOHA(CNA).
Notably, the DOFs of FOHA(CNA) increase at a significantly faster rate compared to the other designs,
highlighting the effectiveness of the enhanced array construction scheme.
The performance advantage of FOHA configurations designed by the proposed scheme
becomes increasingly evident as the number of physical sensors grows.
\begin{figure}
 \center{\includegraphics[width=5.5cm]  {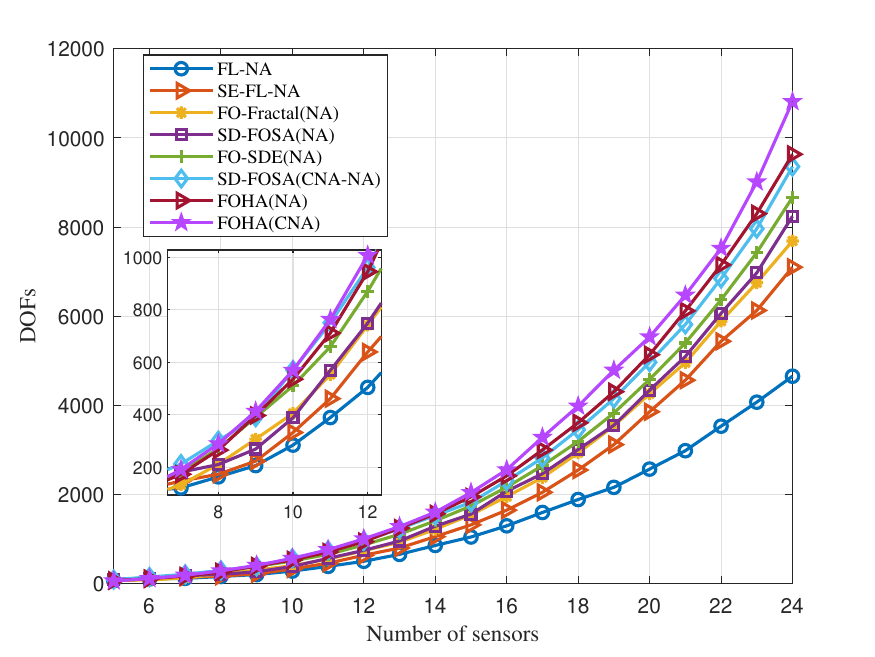}}
 \caption{\label{1} DOFs of different arrays}
 \label{fig:dof}
\end{figure}

\subsection{Comparison of Redundancy for Different Arrays}
We compare the redundancy of the proposed FOHAs with FL-NA, SE-FL-NA, FO-Fractal(NA), FO-SDE(NA), SD-FOSA(NA) and SD-FOSA(CNA-NA) under a fixed number of physical sensors.
Figure~\ref{fig:redundancy} shows the redundancy of these eight arrays as a function of the number of sensors,
along with the theoretical lower bound of fourth-order redundancy.
As the number of sensors \( N \) increases, the redundancy of all arrays increases.
For \( N \leq 8 \), SD-FOSA(CNA-NA) exhibits the lowest redundancy.
However, when \( N \geq 9 \), the redundancy of SD-FOSA(CNA-NA) surpasses that of FOHA(CNA),
and for \( N \geq 13 \), it also exceeds that of FOHA(NA).
Consequently, for \( N \geq 9 \), FOHA(CNA) achieves the lowest redundancy among all arrays while remaining above the theoretical lower bound.
\begin{figure}[h]
 \center{\includegraphics[width=5.5cm]  {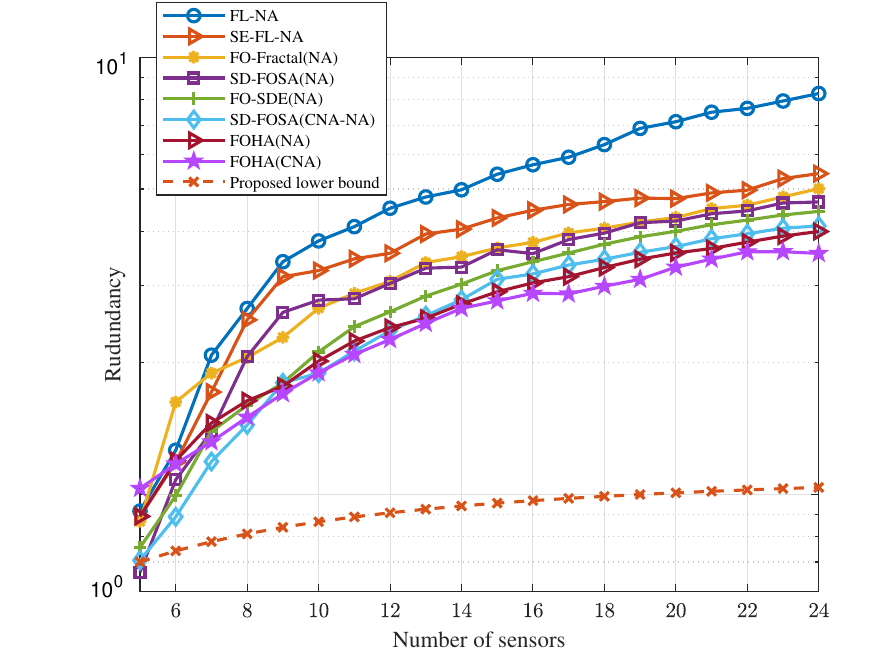}}
 \caption{\label{1} Redundancy of different arrays}
 \label{fig:redundancy}
\end{figure}

\subsection{Mutual Coupling}
This section compares the mutual coupling performance of the proposed FOHAs,
based on different generators, with state-of-the-art co-arrays in terms of coupling leakage.
The mutual coupling model in (\ref{wang5}) is defined by \( c_1 = 0.3 e^{j \pi/3} \), \( B = 100 \),
and \( c_l = c_1 e^{-j (l-1) \pi/8} / l \) for \( 2 \leq l \leq B \).
According to Proposition \ref{propo:L}, the mutual coupling leakage for each array is illustrated in Fig.~\ref{fig:mutual}.
Due to the parameters \(\delta_1, \delta_2\) and \(\eta_1, \eta_2\),
subarrays \(\mathbb{A}_2\) and \(\mathbb{A}_3\) contribute negligibly to the total coupling leakage.
Consequently, the overall coupling leakage is primarily governed by subarray \(\mathbb{A}_1\).
It is evident that, for physical sensors \( N \geq 8 \), the two proposed FOHAs exhibit lower coupling leakage compared to all other arrays considered.
\begin{figure}[h]
 \center{\includegraphics[width=5.5cm]  {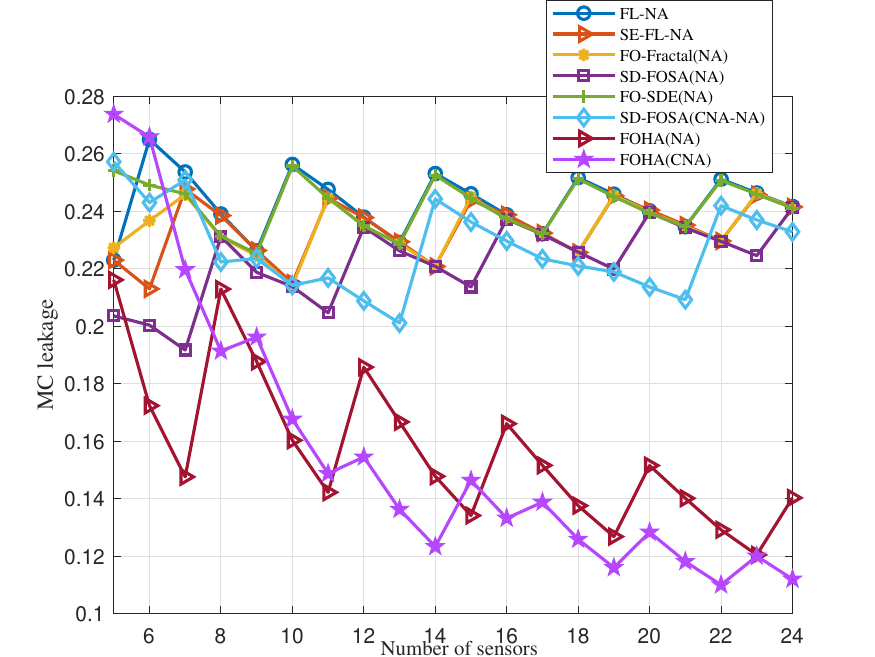}}
 \caption{\label{1} Coupling leakage versus total number of sensors}
 \label{fig:mutual}
\end{figure}

\subsection{Resolution of Different Arrays}
The RMSE respect to the angular separation are shown in Fig. \ref{fig:resolution},
where SNR $=2dB$, $K=8000$ and the number of sources is 20.
It can be seen that smaller angular separation incurs a higher RMSE,
since it is hard to distinguish between two closely spaced sources.
As shown in Fig. \ref{fig:resolution}, the RMSE decreases with the increase of angular separation
until the separation is sufficiently large.
The ranking in estimation accuracy remains consistent with previous simulations,
with the proposed FOHA(NA) and FOHA(CNA) achieving considerable accuracy even at an angular separation of $0.9^{\circ}$ and $1.1^{\circ}$,
which are much better than the other six SLAs.
\begin{figure}[h]
 \center{\includegraphics[width=5.9cm]  {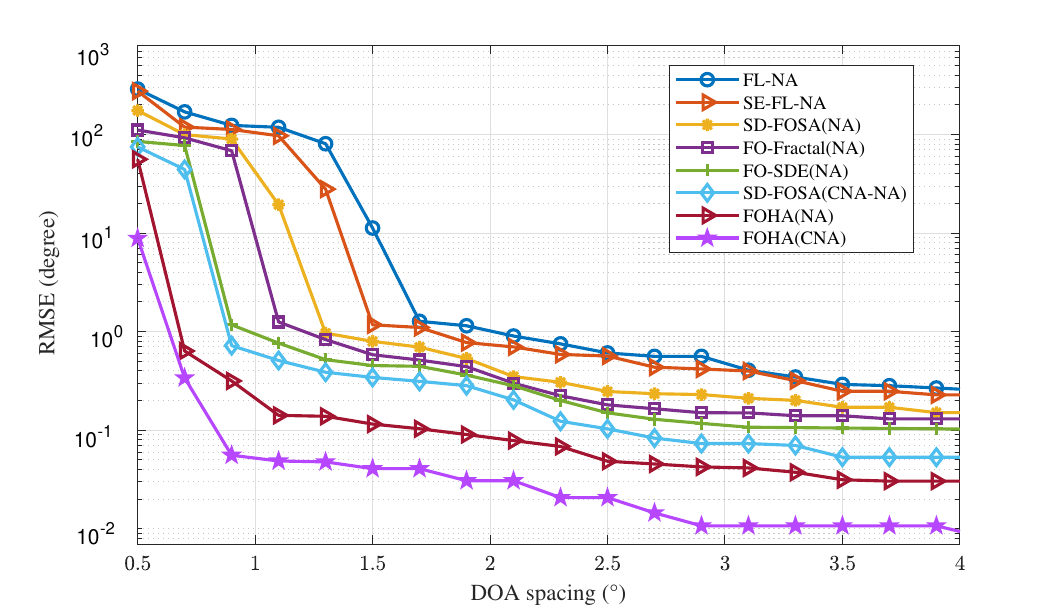}}
 \caption{ RMSE versus angular separation are D=20, SNR=2dB and K=8000}
 \label{fig:resolution}
\end{figure}

\subsection{DOA Estimation Without Mutual Coupling}
\begin{figure}[h]
  \centering
  \subfigure[ RMSE versus SNR]{
    \includegraphics[scale=0.34]{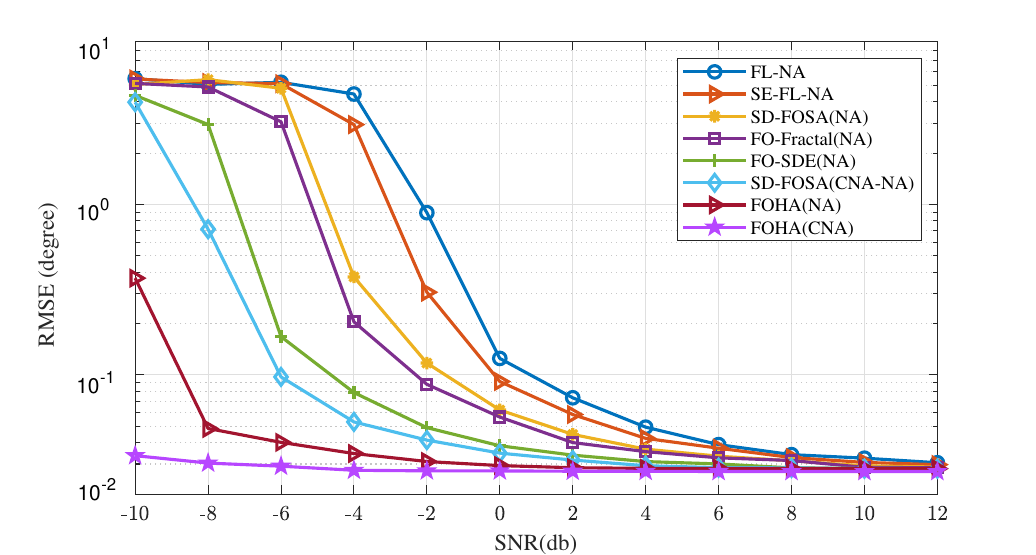}}
   \hspace{0in} 
  \subfigure[RMSE versus Snapshots]{
    \includegraphics[scale=0.34]{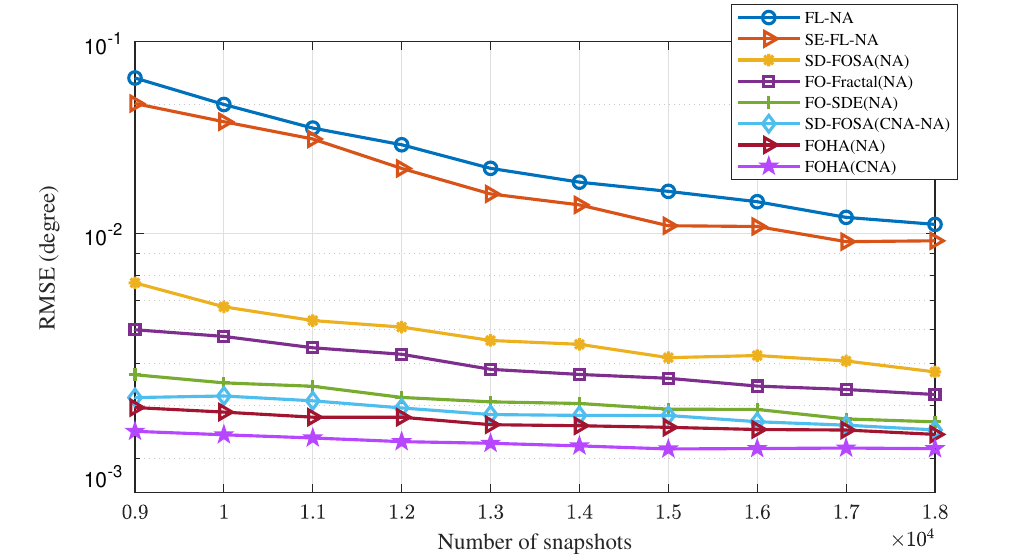}}
    \hspace{0in} 
  \subfigure[RMSE versus The number of sources]{
    \includegraphics[scale=0.34]{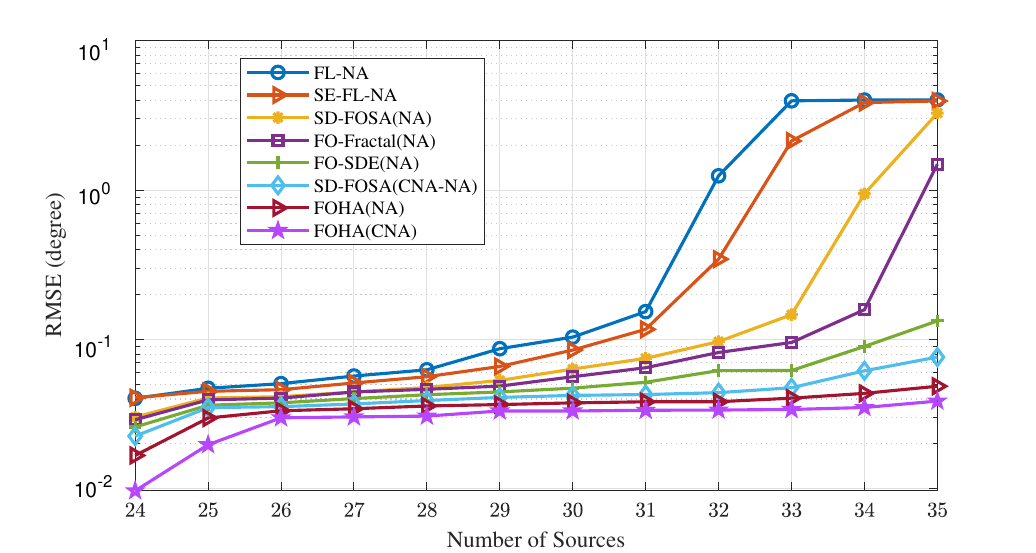}}
  \caption{DOA estimation performance without mutual coupling based on fourth-order cumulants}
  \label{fig:no-coup}
\end{figure}
We compare the DOA estimation performance versus input SNR, number of snapshots and number of sources for eight arrays
without mutual coupling in this part,
where 11 physical sensors are used to construct eight arrays.

In the first simulation, 22 uncorrelated sources are uniformly distributed over \(-85^\circ\) to \(85^\circ\), with 8000 snapshots.
The SNR varies from \(-10\) dB to \(12\) dB in 2 dB increments.
As shown in Fig.~\ref{fig:no-coup}(a), the RMSE of all arrays decreases as the SNR increases.
FOHA(CNA) achieves the lowest RMSE, followed by FOHA(NA), both outperforming the other six SLAs based on fourth-order cumulant.
In the second simulation, the number of snapshots ranges from 9000 to 18000 at a fixed SNR of 2 dB.
Fig.~\ref{fig:no-coup}(b) shows that FOHA(CNA) maintains the lowest RMSE, with FOHA(NA) consistently outperforming the other six arrays.
In the third simulation, the number of sources increases from 24 to 35.
As illustrated in Fig.~\ref{fig:no-coup}(c), FOHA(CNA) exhibits the smallest RMSE in all cases, demonstrating superior performance among the eight arrays.

\subsection{DOA Estimation With Mutual Coupling}
We compare the DOA estimation performance with mutual coupling versus input SNR,
number of snapshots and number of sources for the two FOHAs to
those of six SLAs based on fourth-order cumulant in this part,
where 11 physical sensors are used to construct eight arrays,
and the mutual coupling parameter $B$ is always set to 100 \cite{Liao2012}.

In the first simulation, 22 uncorrelated sources are uniformly distributed between \(-85^\circ\) and \(85^\circ\), with 8000 snapshots.
The SNR varies from \(-10\,\mathrm{dB}\) to \(12\,\mathrm{dB}\) in \(2\,\mathrm{dB}\) steps.
Fig.~\ref{fig:coup}(a) shows the RMSE versus SNR, where the RMSE decreases as SNR increases.
Among all arrays, FOHA(CNA) consistently achieves the lowest RMSE under mutual coupling effects, followed by FOHA(NA),
both outperforming the other six SLAs.
The second simulation investigates the impact of snapshots, varying from 9000 to 18000 at a fixed SNR of \(2\,\mathrm{dB}\).
As shown in Fig.~\ref{fig:coup}(b), FOHA(CNA) maintains the lowest RMSE, with FOHA(NA) again outperforming the other six arrays.
In the third simulation, the number of sources increases from 24 to 35.
Fig.~\ref{fig:coup}(c) shows that FOHA(CNA) consistently yields the smallest RMSE,
demonstrating superior performance across varying source counts compared to the other seven arrays.
\begin{figure}
  \centering
  \subfigure[ RMSE versus SNR]{
    \includegraphics[scale=0.34]{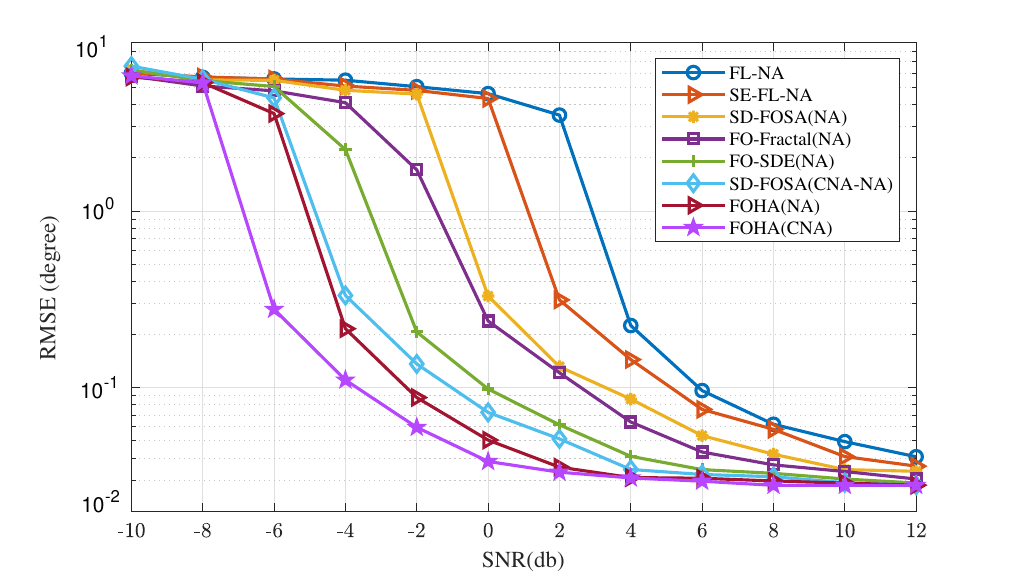}}
  \hspace{0in} 
  \subfigure[RMSE versus Snapshots]{
    \includegraphics[scale=0.34]{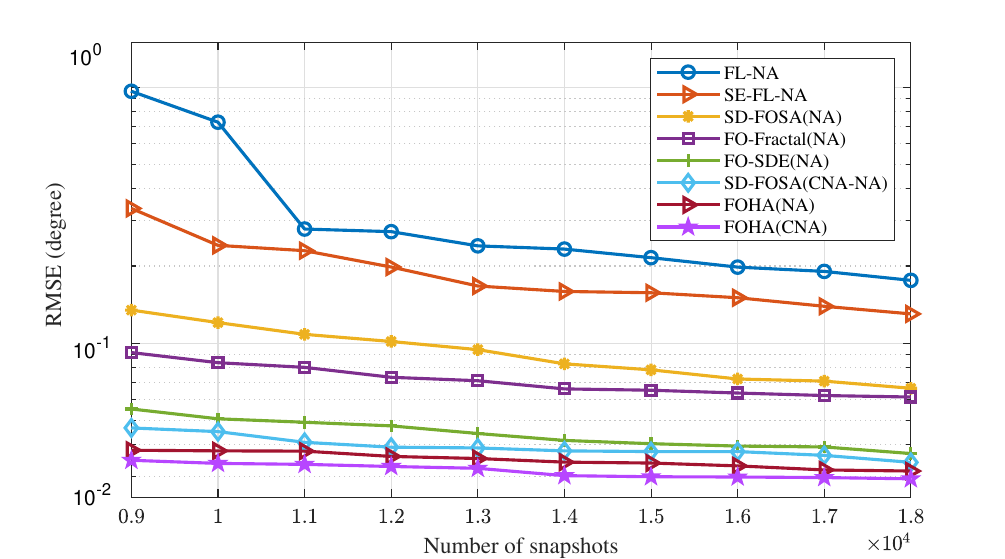}}
  \hspace{0in} 
  \subfigure[RMSE versus The number of sources]{
    \includegraphics[scale=0.34]{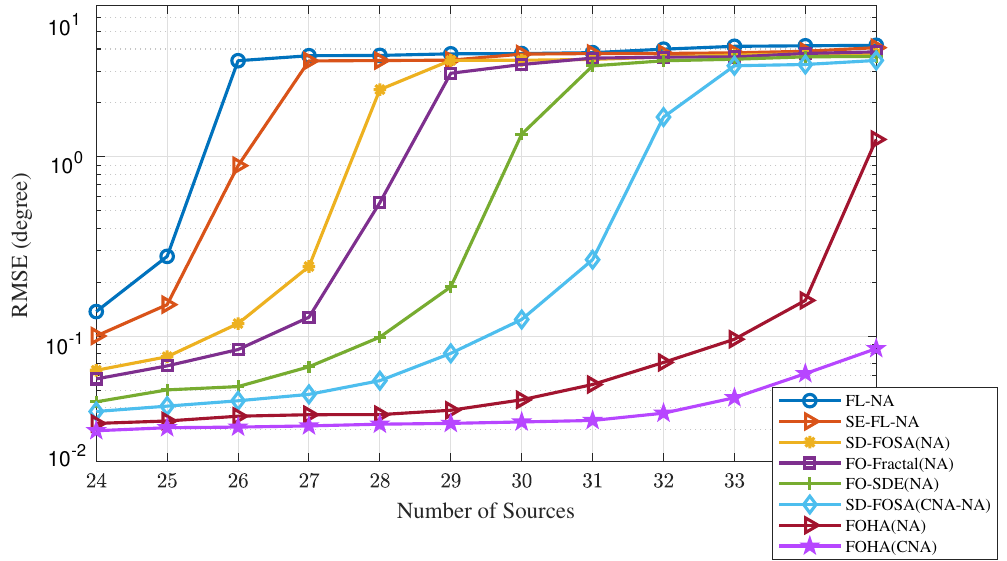}}
  \caption{DOA estimation performance with mutual coupling based on fourth-order cumulants}
  \label{fig:coup}
\end{figure}

\section{Conclusion}
This paper employs high-order statistics, specifically fourth-order cumulants,
to derive FODCA with two different forms.
Based on the resulting FODCA, a novel scheme is introduced for designing fourth-order hierarchical arrays (FOHA)
based on FODCA with two different forms.
Two SLAs, termed FOHA(NA) and FOHA(CNA),
are further proposed using the novel scheme to enhance DOFs
and suppress mutual coupling effects, thereby improving DOA estimation performance.
Closed-form expressions for the sensor locations of FOHA(NA) and FOHA(CNA) are derived through Algorithms~1 and 2, respectively.
Additionally, the concept of fourth-order redundancy is introduced,
and the relationship between the coupling leakage of the generator and
that of the resulting FOHA is derived.
Compared with existing SLAs based on FODCA,
the proposed FOHA(NA) and FOHA(CNA) exhibit reduced mutual coupling and improved spatial resolution.
Nevertheless, due to the increased estimation variance inherent in high-order statistics,
the FOHA structures require more snapshots than SLAs based on second-order statistics.
A promising future direction involves developing DOA estimation algorithms that exploit the FODCA structure more efficiently,
as suggested in~\cite{Kane2024}, to enhance practical performance.

\section*{Acknowledgment}
This work was supported by the National Natural Science Foundation of China (Grant No. 62371400).\\


\begin{thebibliography}{[99]}

\bibitem{Xiao2017notes}
H. S. Xiao and G. Q. Xiao, ``Notes on CRT-based robust frequency estimation,''
\textit{Signal processing.}, vol. 133, pp. 13-17, 2017.

\bibitem{Xiao2018robustness}
H. S. Xiao, Y. F. Huang, Y. Ye and G. Q. Xiao, ``Robustness in chinese remainder theorem for multiple numbers and remainder coding,''
\textit{IEEE Transactions on Signal Processing.}, vol. 66, no. 16, pp. 4347-4361, 2018.


\bibitem{Xiao2016symmetric}
H. S. Xiao, C. Cremers and H. K. Garg, ``Symmetric polynomial $\&$ CRT based algorithms for multiple frequency determination from undersampled waveforms,''
\textit{2016 IEEE Global Conference on Signal and Information Processing (GlobalSIP).}, pp. 202-206, 2016.

\bibitem{Xiao2021wrapped}
H. S. Xiao, N. Du, Z. K. Wang and G. Q. Xiao, ``Wrapped ambiguity Gaussian mixed model with applications in sparse sampling based multiple parameter estimation,''
\textit{Signal Processing.}, vol. 179, pp. 107825, 2021.

\bibitem{Xiao2023on}
H. S. Xiao, Y. W. Zhang, B. N. Zhou and G. Q. Xiao ``On the foundation of sparsity constrained sensing-Part I: Sampling theory and robust remainder problem,''
\textit{IEEE Transactions on Signal Processing.}, vol. 71, pp. 1263-1276, 2023.




\bibitem{Krim1996}
H. Krim and M. Viberg, ``Two decades of array signal processing research: The parametric approach,''
\textit{IEEE Signal Process. Mag.}, vol. 13, no. 4, pp. 67-94, Jul. 1996.

\bibitem{Godara1997}
L. C. Godara, ``Application of antenna arrays to mobile communications. II: Beam-forming and direction-of-arrival estimation,''
\textit{in Proc. IEEE}, vol. 85, no. 8, pp. 1195-1245, Aug. 1997.

\bibitem{Tuncer2009}
T. E. Tuncer and B. Friedlander, ``Classical and Modern Direction-of Arrival Estimation,''
New York, NY, USA: Academic, 2009.

\bibitem{Xiao22023}
H. S. Xiao, J. Wan and S. Devadas, ``Geometry of Sensitivity: Twice Sampling and Hybrid Clipping in Differential Privacy with Optimal Gaussian Noise and Application to Deep Learning,''
\textit{Proceedings of the 2023 ACM SIGSAC Conference on Computer and Communications Security.}, pp. 2636-2650, 2023.

\bibitem{Schmidt1986}
R. O. Schmidt, ``Multiple emitter location and signal parameter estimation,''
\textit{IEEE Trans. Antennas Propag.}, vol. AP-34, no. 3, pp. 276-280, Mar. 1986.


\bibitem{Roy1989}
R. Roy and T. Kailath, ``ESPRIT-estimation of signal parameters via rotational invariance techniques,''
\textit{IEEE Trans. Acoust., Speech, Signal Process.}, vol. 37, no. 7, pp. 984-995, Jul. 1989.

\bibitem{BD2017mutual}
E. BouDaher, F. Ahmad, M. G. Amin, and A. Hoorfar, ``Mutual coupling effect and compensation in non-uniform arrays for direction-of-arrival estimation,''
\textit{Digit. Signal Process.}, vol. 61, pp. 3-14, Feb. 2017.

\bibitem{Pal2010}
P. Pal and P. P. Vaidyanathan, ``Nested arrays: A novel approach to array processing with enhanced degrees of freedom,''
\textit{IEEE Trans. Signal Process.}, vol. 58, no. 8, pp. 4167-4181, Aug. 2010.

\bibitem{Pal2011}
P. P. Vaidyanathan and P. Pal, ``Sparse sensing with co-prime samplers and arrays,''
\textit{IEEE Trans. Signal Process.}, vol. 59, no. 2, pp. 573-586, Feb. 2011.




\bibitem{Xiao2023}
H. S. Xiao, B. N. Zhou, Y. W. Zhang and G. Q Xiao, ``On the foundation of sparsity constrained sensing-Part II: Diophantine sampling with arbitrary temporal and spatial sparsity,''
\textit{IEEE Transactions on Signal Processing.}, vol. 71, pp. 1277-1292, Feb. 2023.

\bibitem{Guo2024}
H. D. Guo, H. Chen, H. G. Lin, W. Liu, Q. Shen and G. Wang, ``A New Fourth-Order Sparse Array Generator Based on Sum-Difference Co-Array Analysis,''
\textit{in ICASSP 2024-2024 IEEE International Conference on Acoustics, Speech and Signal Processing (ICASSP).}, pp. 8486-8490, Mar. 2024.

\bibitem{Moffet1968}
A. Moffet, ``Minimum-redundancy linear arrays,''
\textit{IEEE Trans. Antennas Propag.}, vol. 16, no. 2, pp. 172-175, Mar. 1968.

%


\bibitem{Pal22011}
P. Pal and P. P. Vaidyanathan, ``Coprime sampling and the MUSIC algorithm,''
\textit{in Proc. IEEE Digit. Signal Process. Signal Process. Educ. Meeting.}, pp. 289-294, Mar. 2011.


\bibitem{Zhao2019}
P. Zhao, G. Hu, Z. Qu, and L. Wang, ``Enhanced nested array configuration with hole-free co-array and increasing degrees of freedom for DOA estimation,''
\textit{IEEE Commun. Lett.}, vol. 23, no. 12, pp. 2224-2228, Dec. 2019.

\bibitem{Zheng2019}
Z. Zheng, W. Q. Wang, Y. Kong and Y. D. Zhang, ``MISC array: A new sparse array design achieving increased degrees of freedom and reduced mutual coupling effect,''
\textit{IEEE Trans. Signal Process.}, vol. 67, no. 7, pp. 1728-1741, Apr. 2019.

\bibitem{Shi2022}
W. Shi, Y. Li and R. C. de Lamare, ``Novel sparse array design based on the maximum inter-element spacing criterion,''
\textit{IEEE Signal Process. Lett.}, vol. 29, pp. 1754-1758, 2022.

\bibitem{Liu2015}
C. L. Liu and P. P. Vaidyanathan, ``Remarks on the spatial smoothing step in coarray MUSIC,''
\textit{IEEE Signal Process. Lett.}, vol. 22, no. 9, pp. 1438-1442, Sep. 2015.

\bibitem{Wang2019}
X. Wang and X. Wang, ``Hole identification and filling in k-times extended co-prime arrays for highly efficient DOA estimation,''
\textit{IEEE Trans. Signal Process.}, vol. 67, no. 10, pp. 2693-2706, May. 2019.


\bibitem{Shen2016}
Q. Shen, W. Liu, W. Cui and S. Wu, ``Extension of nested arrays with the fourth-order difference co-array enhancement,''
\textit{in Proc. IEEE Int. Conf. Acoust., Speech, Signal Process.}, pp. 2991-2995, Mar. 2016.

\bibitem{Piya2012}
P. Pal and P. P. Vaidyanathan, ``Multiple level nested array: An efficient geometry for 2qth order cumulant based array processing,''
\textit{IEEE Transactions on Signal Processing.}, vol. 60, no.3, pp. 1253-1269, 2012.


 \bibitem{Nikias1993}
 C. L. Nikias and J. M. Mendel,
 ``Signal processing with higher order spectra,''
 \textit{IEEE Signal Process. Mag.}, vol. 10, no. 3, pp. 10-37, Jul. 1993.

 \bibitem{Dongan1995}
 M. C. Dogan and J. M. Mendel,
 ``Applications of cumulants to array processing-Part I: Aperture extension and array calibration,''
 \textit{IEEE Trans. Signal Process.}, vol. 43, no. 5, pp. 1200-1216, May 1995.

 \bibitem{Gonen1999}
 E. Gonen and J. M. Mendel,
 ``Applications of cumulants to array processing-Part VI: Polarization and direction of arrival estimation with minimally constrained arrays,''
\textit{ IEEE Trans. Signal Process.}, vol. 47,no. 9, pp. 2589-2592, Sep. 1999.


 \bibitem{Chevalier2005}
 P. Chevalier, L. Albera, A. Ferreol and P. Comon,
 ``On the virtual array concept for higher order array processing,''
\textit{ IEEE Trans. Signal Process.}, vol. 53, no. 4, pp. 1254-1271, Apr. 2005.


\bibitem{CaiJJ2017}
J. J. Cai, W. Liu, Z. Ru and Q. Shen,
``An Expanding and Shift Scheme for Constructing Fourth-Order Difference Co-Arrays,''
\textit{IEEE Trans. Signal Process Letters.}, vol. 24, no. 4, pp. 480-484, 2017.





\bibitem{Shen2019}
Q. Shen, W. Liu, W. Cui, S. L. Wu, and P. Pal, ``Simplified and enhanced multiple level nested arrays exploiting high-order difference co-arrays,''
\textit{IEEE Transactions on Signal Processing.}, vol. 67, no. 13, pp.3502-3515, May. 2019.




\bibitem{Cohen2019}
R. Cohen and Y. C. Eldar, ``Sparse fractal array design with increased degrees of freedom,''
\textit{in Proc. IEEE Int. Conf. Acoust., Speech, Signal Process.}, pp. 4195-4199, 2019.

\bibitem{Cohen2020}
R. Cohen and Y. C. Eldar, ``Sparse array design via fractal geometries,''
\textit{IEEE Trans. Signal Process.}, vol. 68, pp. 4797-4812, 2020.


\bibitem{Shen2017}
Q. Shen, W. Liu, W. Cui, S. Wu, Y. D. Zhang and M. G. Amin, ``Focused compressive sensing for underdetermined wideband DOA estimation exploiting high-order difference coarrays,''
\textit{IEEE Signal Process. Lett.}, vol. 24, no. 1, pp. 86-90, Jan. 2017.

\bibitem{Cui2019}
W. Cui, Q. Shen, W. Liu and S. Wu, ``Low complexity DOA estimation for wideband off-grid sources based on re-focused compressive sensing with dynamic dictionary,''
\textit{IEEE J. Sel. Top. Signal Process.}, vol. 13, no. 5, pp. 918-930, Sep. 2019.

\bibitem{Zhou2018}
C. Zhou, Y. Gu, X. Fan, Z. Shi, G. Mao and Y. D. Zhang, ``Direction-of arrival estimation for coprime array via virtual array interpolation,''
\textit{IEEE Trans. Signal Process.}, vol. 66, no. 22, pp. 5956-5971, Nov. 2018.

\bibitem{Liu20172}
C. L. Liu and P. P. Vaidyanathan, ``Maximally economic sparse arrays and cantor arrays,''
\textit{in Proc. IEEE 7th Int. Workshop Comput. Adv. MultiSensor Adaptive Process.}, pp. 1-5, 2017.

\bibitem{Robin2017}
R. Rajamki and V. Koivunen, ``Sparse linear nested array for active sensing.''
\textit{2017 25th European Signal Processing Conference (EUSIPCO).}, IEEE, 2017.
%
%


\bibitem{Krentel1986}
J. B. Hiriart-Urruty, ``From convex optimization to nonconvex optimization. Necessary and sufficient conditions for global optimality,''
\textit{Nonsmooth optimization and related topics.}, pp. 219-239, 1989.

%

%
%


\bibitem{Zhou2020}
Y. Zhou, Y. Li, L. Wang, C. Wen and W. Nie, ``The compressed nested array for underdetermined DOA estimation by fourth-order difference coarrays,''
\textit{in Proc. IEEE Int. Conf. Acoust., Speech, Signal Process., (ICASSP).}, pp. 4617-4621, May. 2020.

\bibitem{LiuJ2017}
J. Liu, Y. Zhang, Y. Lu, S. Ren, and S. Cao, ``Augmented nested arrays with enhanced DOF and reduced mutual coupling,''
\textit{IEEE Trans. Signal Process.}, vol. 65, no. 21, pp. 5549-5563, Nov. 2017.

\bibitem{Friedlander1991}
B. Friedlander and A. J.Weiss, ``Direction finding in the presence of mutual coupling,''
\textit{IEEE Trans. Antennas Propag.}, vol. 39, no. 3, pp. 273-284, Mar. 1991.


\bibitem{Svantesson19991}
T. Svantesson, ``Modeling and estimation of mutual coupling in a uniform linear array of dipoles,''
\textit{in Proc. IEEE Int. Conf. Acoust., Speech, Signal Process.}, pp. 2961-2964, Mar. 1999.

\bibitem{Svantesson19992}
T. Svantesson, ``Direction finding in the presence of mutual coupling,''
\textit{Masters thesis, Dept. Signals Syst., Chalmers Univ. Technol.}, 1999.


\bibitem{Liao2012}
B. Liao, Z. G. Zhang and S. C. Chan, ``DOA estimation and tracking of ULAs with mutual coupling,''
\textit{IEEE Trans. Aerosp. Electron. Syst.}, vol. 48, no. 1, pp. 891-905, Jan. 2012.

%



\bibitem{You2021}
C. Y. You and D. H. Zhang, ``Research on DOA Estimation Based on Spatial Smoothing Algorithm and Optimal Smoothing Times,''
\textit{IEEE International Conference on Data Science and Computer Application (ICDSCA).}, 2021.

\bibitem{Kane2024}
D. M. Kane, I. Diakonikolas, H. S. Xiao and S. H. Liu, ``Online Robust Mean Estimation,''
\textit{Proceedings of the 2024 Annual ACM-SIAM Symposium on Discrete Algorithms (SODA). Society for Industrial and Applied Mathematics.}, pp. 3197-3235, 2024.

\bibitem{DengJ2020}J. Deng, Q. Wang, J. Xie,
 ``Fourth-order cumulant based direction finding algorithm for non-circular signals using uniform circular array with mutual coupling,''
 \textit{2020 IEEE International Conference on Signal Processing, Communications and Computing (ICSPCC). IEEE}, 2020: 1-5.

%
%
%
%


\bibitem{Leech1956}
J. Leech, ``On the representation of 1,2,...,n by differences,''
\textit{J. London Math. Soc.}, vol. s1-31, no. 2, pp. 160-169, 1956.

\bibitem{ChenYP2021}
Y. P. Chen and C. L. Liu,
``On the size and redundancy of the fourth-order difference co-array,''
\textit{IEEE Signal Processing Letters.}, 2021, 28: 2013-2017.

\bibitem{ChenYP2023}
Y. P. Chen and C. L. Liu,
``Half-Inverted Array Design Scheme for Large Hole-Free Fourth-Order Difference Co-Arrays,''
\textit{IEEE Signal Processing.}, vol. 71, pp. 3296-3312.

\bibitem{LiuJHZ2007}
J. Liu, Z. Huang and Y. Zhou,
``A direction finding algorithm for noncircular signals based on higher-order cumulants, 2007 International Conference on Wireless Communications,''
\textit{Networking and Mobile Computing.}, 2007. pp. 1044-1047.

\bibitem{WangB2020}
B. Wang and J. Zheng,
``Cumulant-based DOA estimation of noncircular signals against unknown mutual coupling,''
\textit{Sensors} 20(3) (2020) 878.

%
%
%

 \bibitem{YangZ2022}
Z. Yang, Q. Shen, W. Liu, and W. Cui, 
``A sum-difference expansion scheme for sparse array construction based on the fourth-order difference co-array,''
\textit{IEEE Signal Processing Letters}, vol. 29, pp. 2647-2651, 2022.
 \bibitem{ChenH2025}
H. Chen, H. D. Guo, W. Liu, Q. Shen G. Wang and H. Cheung So,
``Fourth-Order Sparse Array Design from a Sum-Difference Co-array Perspective,''
\textit{IEEE Transaction on Signal Processing}, 2025.

 \bibitem{ThompsonAR2017}
A. R. Thompson, J. M. Moran, and G. W. Swenson,
\textit{Interferometry and Synthesis in Radio Astronomy},
3rd ed. Cham, Switzerland: Springer, 2017.

\bibitem{Horn2012}
R. A. Horn and C. R. Johnson.
\textit{Matrix analysis},
Cambridge university press, 2012.

\bibitem{Yang2023}
Z. X. Yang, Q. Shen, W. Liu, Y. C. Eldar and W. Cui, 
``High-order cumulants based sparse array design via fractal geometries Part I: Structures and DOFs,''
\textit{IEEE Transactions on Signal Processing.}, vol. 71, pp. 327-342, Feb. 2023.





\end{thebibliography}
\end{document}